\documentclass[10pt,english,a4paper]{article}

\makeatletter{}
\usepackage{ifpdf}   \usepackage{multido} \usepackage{ifthen}

\usepackage{color}

\ifpdf
\usepackage[pdftex]{graphicx}
\else
\usepackage[dvips]{graphicx}
\fi

\usepackage{float,floatflt,afterpage}
\usepackage{wrapfig}
\usepackage{placeins}

\usepackage{etex}
\usepackage{tikz}
\usetikzlibrary{snakes}

\usepackage[arrow,matrix,curve]{xy} \usepackage{subfigure} 

\graphicspath{{../figures/}}

\usepackage[latin1]{inputenc}
\usepackage[T1]{fontenc}

\usepackage[centertags,intlimits]{amsmath}
\usepackage{amsthm}
\usepackage{amssymb}
\usepackage{mathtools}
\usepackage{mathrsfs}

\theoremstyle{break}
\allowdisplaybreaks[1]

\usepackage{exscale}
\usepackage{dsfont}

\usepackage[numbers,square,sort&compress]{natbib}

\ifpdf
\usepackage[pdftex,
  bookmarks,
  bookmarksopen=true,
  bookmarksnumbered=true,
  backref,
  linkcolor=blue,
  urlcolor=blue,
  pdfauthor={Andreas Fischle},
  pdftitle={Optimality of the relaxed polar factors by a characterization
  of the set of real square roots of real symmetric matrices}
  hyperfigures=true,
  pdfpagelabels,
  pagebackref,
  hypertexnames=true,
  plainpages=false,
  naturalnames]{hyperref}\else
\usepackage[dvips,breaklinks]{hyperref}                                         \usepackage{breakurl}
\fi

\usepackage{hypernat}

\usepackage{url} \usepackage{breakurl} 

\makeatletter
\def\url@leostyle{  \@ifundefined{selectfont}{\def\UrlFont{\sf}}{\def\UrlFont{\small\ttfamily}}}
\makeatother
\usepackage{enumerate}
\usepackage{verbatim}        \usepackage{lscape}          \usepackage{array,longtable} \usepackage{listings}        \usepackage{caption}

\usepackage{makeidx}

\makeatletter{}

                                                                                                                                                                                                                                                                                                            \def\and{{\rm and}}            \def\det#1{{\rm det}\,{#1}}            
          \def\exp{{\rm exp}}                                        
\def\skew{{\rm skew}}                                  
\newfont{\Sf}{cmssbx10 scaled 2074}
\newbox{\assem}
\savebox{\assem}(15,2)[bl]{
\begin{normalsize}
\unitlength=1.0mm
\begin{picture}(6.00,0.00)
\put(.0,1.0){\makebox(0,0)[cc]{$\cup$}}
\put(.0,-1.0){\makebox(0,0)[cc]{$\mbox{\scriptsize \it n=1}$}}
\put(.0,4.0){\makebox(0,0)[cc] {\scriptsize $\it n_{\mbox{\tiny\it elm}}$}}
\end{picture}
\end{normalsize}}

\newbox{\asse}
\savebox{\asse}(15,2)[bl]{
\begin{normalsize}
\unitlength=1.0mm
\begin{picture}(8.00,0.00)
\put(.0,2.0){\makebox(0,0)[cc]{\Sf A}}
\end{picture}
\end{normalsize}}

                                                        \def\sqtwo3{{\textstyle {\sqrt{2 \over 3}}}}   \newcommand{\IP}{{\rm I\kern-.18em P}}           \newcommand{\II}{{\rm I\kern-.18em I}}           \newcommand{\IF}{{\rm I\kern-.25em F}}           \newcommand{\IE}{{\rm I\kern-.25em E}}           \def\IR{{\rm I\kern-.15em R}}

\newcommand{\ia}{{\rm\kern.24em                     \vrule width.02em height0.9ex depth-.05ex
   \kern-.26em a}}
\newcommand{\ic}{{\rm\kern.24em                     \vrule width.02em height0.9ex depth-.05ex
   \kern-.26em c}}
\newcommand{\IA}{{\rm\kern.22em                      \vrule width.02em
        height0.5ex depth 0ex
    \kern-.24em A}}
\newcommand{\IC}{{\rm\kern.24em                     \vrule width.02em height1.4ex depth-.05ex
   \kern-.26em C}}

\makeatletter{}

\DeclareMathOperator{\dist}{dist}

\renewcommand{\epsilon}{\varepsilon}

\newcommand{\dV}{\,{\rm dV} }

\newcommand{\norm}[1]{\|#1\|}
\newcommand{\abs}[1]{\left| #1 \right|}

\makeatletter
\newtheorem*{rep@theorem}{\rep@title}
\newcommand{\newreptheorem}[2]{\newenvironment{rep#1}[1]{ \def\rep@title{#2 \ref{##1}} \begin{rep@theorem}} {\end{rep@theorem}}}
\makeatother
\newreptheorem{theo}{Theorem}

\newtheorem{lem}{Lemma}[section]
\newtheorem{ass}[lem]{Assumption}

\newtheorem{rem}[lem]{Remark}
\newtheorem{defi}[lem]{Definition}

\newtheorem{theo}[lem]{Theorem}

\newtheorem{cor}[lem]{Corollary}
\newtheorem{scheme}[lem]{Scheme}
\newtheorem{exa}[lem]{Example}
\newtheorem{prob}[lem]{Problem}

\newcommand{\leref}[1]{Lemma \ref{#1}}
\newcommand{\assref}[1]{Assumption \ref{#1}}
\newcommand{\theref}[1]{Theorem \ref{#1}}

\newcommand{\coref}[1]{Corollary \ref{#1}}
\newcommand{\exaref}[1]{Example \ref{#1}}
\newcommand{\remref}[1]{Remark \ref{#1}}

\newcommand{\probref}[1]{Problem \ref{#1}}

\newcommand{\R}{\mathbb{R}}
\newcommand{\N}{\mathbb{N}}
\newcommand{\C}{\mathbb{C}}

\DeclareMathOperator{\GL}{GL}
\DeclareMathOperator{\SO}{SO}
\let\O\relax
\DeclareMathOperator{\O}{O}

\DeclareMathOperator{\skewop}{skew}
\renewcommand{\skew}{\skewop}
\DeclareMathOperator{\diag}{diag}
\DeclareMathOperator{\sym}{sym}
\DeclareMathOperator{\Tr}{tr}

\DeclareMathOperator{\dev}{dev}

\DeclareMathOperator{\so}{\mathfrak{so}}

\DeclareMathOperator{\polar}{R_{\rm p}}

\newcommand{\Sym}{ {\rm{Sym}} }
\newcommand{\Psym}{ {\rm{PSym}} }

\newcommand{\id}{{\boldsymbol{\mathds{1}}}}
 
\DeclareMathOperator{\Det}{det}

\renewcommand{\det}[1]{ {\Det[{#1}]} }
\newcommand{\tr}[1]{ {\Tr \left[{#1}\right]} }

\makeatletter{}

\newcommand{\secref}[1]{Section \ref{#1}}

\definecolor{orange}{rgb}{1.0,0.5,0}

\DeclareMathOperator{\Reals}{\mathbb{R}}
\renewcommand{\R}{\Reals}

\DeclareMathOperator{\argminmathop}{arg\,min}
\newcommand{\argmin}[2]{\mathchoice{\underset{#1}{\argminmathop}\, {#2}}{\argminmathop_{#1}\, {#2}}{}{}}

\newcommand{\scalprod}[2]{\big<#1,\,#2\big>}

\newcommand{\setdef}[2]{\lbrace #1 \;\vert\; #2\rbrace}

\DeclareMathOperator{\spanop}{span}

\newcommand{\vspan}[1]{\spanop\left(\left\{ #1 \right\}\right)}

\newcommand{\hsnorm}[1]{\left\lVert #1 \right\rVert}

\DeclareMathOperator{\Skew}{Skew}
\DeclareMathOperator{\Diag}{Diag}

\DeclareMathOperator{\RPosZ}{\sideset{}{_0^+}\Reals}

\newcommand{\eqdef}{\coloneqq}

\newcommand{\eqiso}{\cong}
\newcommand{\isequivto}{\Longleftrightarrow}

\makeatletter{}

\newcommand{\mrot}{\overline{R}}
\newcommand{\mstretch}{\overline{U}}

\DeclareMathOperator{\rpolar}{rpolar}

\DeclareMathOperator{\wmm}{W_{\mu,\mu_c}}

\newcommand{\countres}{
  \setcounter{equation}{0}
  \setcounter{figure}{0}
  \setcounter{table}{0}
}

                     \renewcommand{\baselinestretch}{1.0}          

\makeatletter

\renewcommand{\itemize}{  \ifnum \@itemdepth >\thr@@\@toodeep\else
    \advance\@itemdepth\@ne
    \edef\@itemitem{labelitem\romannumeral\the\@itemdepth}    \expandafter
    \list
      \csname\@itemitem\endcsname
      {\def\makelabel##1{\hss\llap{##1}}        \topsep=.8ex\itemsep=-.2ex}  \fi}

\renewcommand\section{\@startsection {section}{1}{\z@}  {-3.5ex \@plus -1ex \@minus -.2ex}  {2.3ex \@plus.2ex}  {\boldmath\normalfont\Large\bfseries}}
\renewcommand\subsection{\@startsection{subsection}{2}{\z@}  {-3.25ex\@plus -1ex \@minus -.2ex}  {1.5ex \@plus .2ex}  {\boldmath\normalfont\large\bfseries}}
\renewcommand\subsubsection{\@startsection{subsubsection}{3}{\z@}  {-3.25ex\@plus -1ex \@minus -.2ex}  {1.5ex \@plus .2ex}  {\boldmath\normalfont\normalsize\bfseries}}
\renewcommand\paragraph{\@startsection{paragraph}{4}{\z@}  {3.25ex \@plus1ex \@minus.2ex}  {-1em}  {\boldmath\normalfont\normalsize\bfseries}}
\renewcommand\subparagraph{\@startsection{subparagraph}{5}{\parindent}  {3.25ex \@plus1ex \@minus .2ex}  {-1em}  {\boldmath\normalfont\normalsize\bfseries}}

\makeatother

\makeatletter
\renewcommand*{\@fnsymbol}[1]{\ensuremath{\ifcase#1\or 1\or 2\or 3\or
   \mathsection\or \mathparagraph\or \|\or **\or \dagger\dagger
   \or \ddagger\ddagger \else\@ctrerr\fi}}
\makeatother

\title{Optimality of the relaxed polar factors by a characterization of
       the set of real square roots of real symmetric matrices}

\author{Lev Borisov
\!\thanks{Lev Borisov,
Department of Mathematics,
Rutgers University, 240 Hill Center, Newark, NJ 07102, United States,
email: borisov@math.rutgers.edu},
\quad Andreas Fischle
\!\thanks{Corresponding author: Andreas Fischle,
Institut f\"ur Numerische Mathematik,
TU Dresden,
Zellescher Weg 12-14,
01069 Dresden,
Germany,
email: andreas.fischle@tu-dresden.de},
\quad and \quad
Patrizio Neff
\!\thanks{Patrizio Neff,
Head of Lehrstuhl f\"{u}r Nichtlineare Analysis und Modellierung,
Fakult\"{a}t f\"{u}r Mathematik,
Universit\"{a}t Duisburg-Essen,
Thea-Leymann Str. 9,
45127 Essen,
Germany,
email: patrizio.neff@uni-due.de}}

\begin{document}
\selectfont
\maketitle

\pagenumbering{arabic}

\makeatletter{}\begin{center}\textbf{Abstract}\end{center}
\begin{center}
  \begin{minipage}{0.95\textwidth}
  We consider the problem to determine the optimal rotations
    $R \in \SO(n)$ which minimize
    \begin{align*}
      W: \SO(n) \to \RPosZ,
      \quad
      W(R\,;D) \;\eqdef\; \hsnorm{\sym(R D - \id)}^2
    \end{align*}
    for a given diagonal matrix
    $D \eqdef \diag(d_1, \ldots, d_n) \in \R^{n \times n}$
    with positive entries $d_i > 0$.
    The objective function $W$ is the reduced form of
    the Cosserat shear-stretch energy, which, in its general form,
    is a contribution in any geometrically nonlinear, isotropic,
    and quadratic Cosserat micropolar (extended) continuum model.
    We characterize the critical points of the energy $W(R\,;D)$,
    determine the global minimizers and compute the global minimum.
    This proves the correctness of  previously obtained formulae
    for the optimal Cosserat rotations in dimensions two and three.
    The key to the proof is the result that every real matrix whose
    square is symmetric can be written in some
    orthonormal basis as a block-diagonal matrix with blocks
    of size at most two. This statement does not seem to
    appear in the literature.
\end{minipage}
\end{center}

\vspace*{0.125cm}
{\small
  {\bf{Keywords:}}
  Cosserat theory,
  micropolar media,
  Grioli's theorem,
  rotations,
  special orthogonal group,
  (non-symmetric) matrix square root,
  symmetric square,
  polar decomposition,
  relaxed-polar decomposition.

{\bf{AMS 2010 subject classification:}}
  15A24,     22E30,     74A30,     74A35,     74B20,     74G05,     74G65,     74N15.   }

 \countres

\setcounter{tocdepth}{1}
\renewcommand{\baselinestretch}{-1.0}\normalsize
{
  \small
  \tableofcontents
}
\renewcommand{\baselinestretch}{1.0}\normalsize

\makeatletter{}\section{Introduction}
\label{sec:intro}

\subsection{The problem}
In this contribution, we characterize the solutions to the optimality
problem stated as
\begin{prob}
  \label{prob:opt}
  Let $D \eqdef \diag(d_1,\ldots,d_n) > 0$ be a positive definite
  diagonal matrix and let
  \begin{equation}
    \label{eq:def_energy}
    W: \SO(n) \,\times\, {\rm Diag}(n) \to \RPosZ,
    \quad
    W(R\,;D) \;\eqdef\; \hsnorm{\sym(RD - \id)}^2\;.
  \end{equation}
  Compute the relaxed polar factors, i.e., the set of energy-minimizing
  rotations
  \begin{equation}
    \rpolar(D) \;\eqdef\;
    \argmin{R\,\in\,\SO(n)}{W(R\,;D)}
    \;=\;
    \argmin{R\,\in\,\SO(n)}{\hsnorm{\sym(RD - \id)}^2} \;\subseteq\; \SO(n)\;.
  \end{equation}
\end{prob}
We use the notation
$\sym(X)  \eqdef \frac{1}{2}(X + X^T)$,
$\skew(X) \eqdef \frac{1}{2}(X - X^T)$,
$\dev(X)  \eqdef X - \frac{1}{n}\, \tr{X} \cdot \id$,
$\scalprod{X}{Y} \eqdef \tr{X^TY}$ and we denote the induced
Frobenius matrix norm by
$\hsnorm{X}^2 \eqdef \scalprod{X}{X} = \sum_{1 \leq i,j \leq n} X_{ij}^2$.
We call a rotation $R \in \SO(n)$ optimal for given $D \in \Diag(n)$
if it is a global minimizer for the energy $W(R\,;D)$ defined
in~\eqref{eq:def_energy}. Furthermore, we denote the spaces of
symmetric and skew-symmetric matrices by $\Sym(n) \subset \R^{n \times n}$
and $\Skew(n) \subset \R^{n \times n}$, respectively.

This work is concerned with the derivation of formulae that explicitly
characterize the relaxed polar factors $\rpolar(D)$ in arbitrary
dimension. It is beyond the scope of the current paper to develop
efficient and stable numerical approximations of the relaxed polar
factors $\rpolar(D)$. However, this will be a logical next step.

\subsection{Previous results}
We consider the quadratic Cosserat shear-stretch energy
$W_{\mu,\mu_c}: \SO(n) \times \GL^+(n) \to \RPosZ$
\begin{equation}
  \label{eq:intro:wmm}
  W_{\mu,\mu_c}(\mrot\,;F) \eqdef
  \mu\,\hsnorm{\sym(\mrot^TF - \id)}^2
  \,+\,
  \mu_c\,\hsnorm{\skew(\mrot^TF - \id)}^2
\end{equation}
with weights (material parameters) $\mu > 0$ and
$\mu_c \geq 0$.\footnote{A more in-depth presentation of the
  interpretation in mechanics is provided in~\secref{subsec:mechanics}.}
Let us introduce the general weighted form of the relaxed polar
factors
\begin{equation}
  \rpolar_{\mu,\mu_c}(F) \eqdef \argmin{\mrot\,\in\,\SO(n)}
         {\left(\mu\,\hsnorm{\sym(\mrot^TF - \id)}^2
           \,+\,
          \mu_c\,\hsnorm{\skew(\mrot^TF - \id)}^2\right)}\;.
\end{equation}
The unique global minimizer $\mrot \in \SO(n)$ in the \emph{classical
parameter range} $\mu_c \geq \mu > 0$ is the orthogonal factor
$\polar(F)$ in the right polar decomposition of $F \in \GL^+(n)$,
see~\cite{Neff_Grioli14}. The \emph{non-classical parameter range}
$\mu > \mu_c \geq 0$ of parameters can, surprisingly, be reduced
to a single \emph{non-classical limit case}
$(\mu,\mu_c) = (1,0)$, see~\cite{Fischle:2015:OC2D}.
The choice of weights (material parameters) $\mu > 0$
and $\mu_c \geq 0$ is crucial since they characterize
a pitchfork bifurcation between a classical branch of
minimizers, i.e, where $\polar(F)$ is optimal, and an
interesting new type of non-classical minimizers.

Due to the parameter reduction, it suffices to consider the
Cosserat shear-stretch energy in the limit case $(\mu,\mu_c) = (1,0)$
given by
\begin{equation}
  W_{1,0}(\mrot\,;F) \eqdef \hsnorm{\sym(\mrot^TF - \id)}^2\;.
\end{equation}
A Cosserat strain energy $W(F, \mrot)$ is called isotropic, if
it satisfies the invariance $W(Q_1FQ_2, Q_1\mrot Q_2) = W(F,\mrot)$
for all $Q_1,Q_2 \in \SO(n)$. Exploiting the isotropy of the
Cosserat shear-stretch energy, it can be equivalently expressed
in terms of a rotation $R \in \SO(n)$ acting relative to the
polar factor $\polar(F)$, see the introduction
to~\cite{Fischle:2015:OC3D} for details. After this second
reduction step, we obtain the equivalent energy
\begin{equation}
  W_{1,0}(R,;D) \eqdef \hsnorm{\sym(RD - \id)}^2\;,
\end{equation}
where $D = \diag(d_1,\ldots,d_n) > 0$ is a positive definite
matrix. Its diagonal entries are given by the singular values
$d_i = \nu_i > 0$ of $F \in \GL^+(n)$.

Hence, on the basis of the previous works \cite{Fischle:2015:OC2D}
and~\cite{Fischle:2015:OC3D}, it suffices to solve~\probref{prob:opt}
in order to characterize the global minimizers for the quadratic
Cosserat shear-stretch energy in the entire non-classical parameter
range. For a short overview of the previous results,
see \cite{Fischle:2017:GTW}.

Explicit formulae for the critical points and the global minimizers
$\rpolar^\pm_{\mu,\mu_c}(F)$ of $W_{\mu,\mu_c}(\mrot\,;F)$ in dimension
two have been presented in~\cite{Fischle:2015:OC2D}. The
corresponding minimal energy levels were also provided.
In dimension three, the following explicit formulae for the
solutions to \probref{prob:opt} were obtained using computer
algebra~\cite[Corollary 2.7]{Fischle:2015:OC3D}:
\begin{cor}[Energy-minimizing relative rotations for $(\mu,\mu_c) = (1,0)$]
  \label{cor:rhat10}
  Let $D = \diag(d_1,d_2,d_3)$ such that $d_1 > d_2 > d_3 > 0$.
  Then the solutions to~\probref{prob:opt} are given by the
  energy-minimizing relative rotations
  \begin{equation}
    \label{eq:rpolar}
    \rpolar(D) = \left\lbrace
    \begin{pmatrix}
      \cos \alpha  & -\sin \alpha & 0\\
      \sin \alpha  &  \cos \alpha & 0\\
      0            &  0              & 1\\
    \end{pmatrix}
    \right\rbrace\;,
  \end{equation}
  where $\alpha \in [-\pi,\pi]$ is an optimal rotation angle
  satisfying
  \begin{equation}
    \label{eq:alpha}
    \alpha =
    \begin{cases}
      \; 0\;,  &\quad\text{if}\quad d_1 + d_2 \leq 2\;,\\
      \; \pm\arccos(\frac{2}{d_1 + d_2})\;,
      &\quad\text{if}\quad d_1 + d_2 \geq 2\;.
    \end{cases}
  \end{equation}
  In particular, for $d_1 + d_2 \leq 2$, we have $\rpolar(D) = \{\id\}$.
\end{cor}
Note that the validation of the minimizers in dimension three, i.e.,
of the formulae \eqref{eq:rpolar} and \eqref{eq:alpha}
in~\cite{Fischle:2015:OC3D} was based on brute force stochastic
minimization, since a proof of optimality was out of reach.

With the present contribution, we close this gap in $n = 3$ and
generalize the previously obtained formulae $\rpolar^\pm_{1,0}(F)$
from~\cite{Fischle:2015:OC3D,Fischle:2016:RPNI} to arbitrary dimension
$n$. Note that the parameter transformation proved in
\cite{Fischle:2015:OC2D} allows to recover the general solution
in the non-classical parameter range $\rpolar^\pm_{\mu,\mu_c}(F)$
from $\rpolar^\pm_{1,0}(F)$ by a rescaling of the deformation gradient,
but we shall not detail this here.

Our main result is that~\probref{prob:opt} has $2^k$ global minimizers
that are block-diagonal, similar to the $n = 3$ case above. Here,
$k$ is the number of blocks of size two. More precisely, we prove
\begin{reptheo}{theo:global_min_nd}
  Let $D \eqdef \diag(d_1,\ldots,d_n) > 0$ with ordered entries
  $d_1 > d_2 > \ldots > d_n > 0$. Let us fix the maximum $k \in \N_0$
  for which $d_{2k-1} + d_{2k} > 2$. Any global minimizer $R \in \SO(n)$ of
  $$
  W(R\,;D) \eqdef \hsnorm{\sym(RD - \id)}^2
  $$
  corresponds to a partition of the index set $\{1,\ldots,n\}$ with
  $k \geq 0$ leading subsets of size two
  $$
  \underbrace{\{1,2\} \sqcup \{3,4\} \sqcup \ldots \sqcup \{2k-1, 2k\}}_{k\;\rm{subsets\,of\,size\,two}}
    \;\sqcup\;
  \underbrace{\{2k+1\} \sqcup \ldots \sqcup \{n\}}_{(n - 2k)\;\text{subsets\,of\,size one}}
  $$
  in the classification of critical points provided by \theref{theo:critical_values}. The global minimum of $W(R\,;D)$ is given by
\begin{equation}
  W^{\rm red}(D) \eqdef \min_{R \in \SO(n)}{W(R\,;D)}
  = \frac{1}{2}\sum_{i=1}^k (d_{2i-1} - d_{2i})^2 + \sum_{i=2k+1}^n (d_i-1)^2\;.
\end{equation}
\end{reptheo}

We note in passing that the case of optimal rotations for recurring
parameter values $d_i$, $i = 1,\ldots,n$, in the diagonal parameter
matrix $D \in \Diag(n)$ has not been treated previously and is
accessible with the present approach.\footnote{This allows to
  treat special cases of equal principal stretches $\nu_i$ of
  the deformation gradient $F \in \GL^+(n)$
  which may arise, e.g., due to symmetry assumptions.}

\subsection{Algebraic solution strategy and state of the art}
\label{subsec:approach}
Let us lay out our solution strategy for~\probref{prob:opt}
and present the algebraic techniques which lie at the heart
of it. The Euler-Lagrange equations for the function
$W(R\,;D)$ have been previously derived in~\cite{Fischle:2015:OC3D}
and~\cite{Neff_Biot07}. These equations characterize the
critical points of the objective function $W(R\,;D)$
implicitly as solutions of a quadratic matrix equation
posed on the manifold of rotations $\SO(n)$ and parametrized
by the diagonal matrix $D = \diag(d_1,\ldots,d_n)$. The foundation
of our solution approach is the successful analysis of the following
equivalent algebraic condition
\begin{equation}
  (RD - \id)^2 \in \Sym(n)\;.
\end{equation}
This is a \emph{symmetric square condition}
\begin{equation}
  \left(X(R)\right)^2 \;=\; S \in \Sym(n),
  \quad\quad\text{where}\quad\quad
  X(R) \eqdef RD - \id \in \R^{n \times n}\;.
\end{equation}
Given this condition, one might suspect that the computation
of critical points of $W(R\,;D)$ is related to the theory of
real matrix square roots of real symmetric matrices. However,
the delicate properties of matrix square roots can, for the
most part, be avoided and we consider this an intriguing
aspect of our solution approach.

Matrix square roots are a classical theme in matrix analysis.
The most-well known example is the unique symmetric positive
definite, so-called principal, matrix square root of a
symmetric positive definite real square matrix. However, as
it turns out, a given real square matrix can have isolated and
non-isolated families of matrix square roots which can be real,
but are complex in general. A classification of all complex
square roots of a given complex matrix $A \in \C^{n \times n}$
in terms of its Jordan decomposition has been given by
Gantmacher, see~\cite{Gantmacher:1959:MT1}.
This classification can also be adapted to the real case, see,
e.g. Higham~\cite{Higham:1987:CRSQ}. Further treatments and
results on matrix square roots are given in the extensive
monographs~\cite{Horn85,Horn91,Higham:2008:FOM},
while \cite{Gallier:2008:LAS} provides a compact recent
introduction.

Our development is, however, more intuitively phrased
in terms of matrix squares, since we do not rely on the classical
theory of matrix square roots. For example, the key to
the analysis of \probref{prob:opt} is our \theref{theo:real_square_roots}
which states:
  \emph{every real matrix $X \in \R^{n\times n}$ whose
square $S = X^2$ is symmetric can be written in some
orthonormal basis as a block-diagonal matrix with blocks
of size at most two.} This statement does not seem to
appear in the literature. The construction of this
orthonormal basis was originally inspired by the theory
of principal angles between linear subspaces, see,
e.g.~\cite{Galantai:2013:PPM}. We emphasize
that the \emph{orthogonality} of the associated change of basis
matrix $T \in \O(n)$ is of utmost importance for our solution
strategy. We require the change of basis to preserve the
Frobenius matrix norm in~\probref{prob:opt}.
The block-diagonal structure in the new basis allows to break the
minimization problem down into subproblems of dimension at most
two. For example, we shall see that in $n\!=\!3$, for a
non-classical minimizer, we have to solve a one-dimensional and a
two-dimensional subproblem. The one-dimensional problem determines
the rotation axis of the optimal rotations, while the
two-dimensional subproblem determines the optimal rotation angles.

Let us briefly illuminate two similar constructions, which are
however insufficient for our purposes. We expand on these approaches
in the text. First, based on the characterization of the set of
complex matrix square roots due to Gantmacher~\cite{Gantmacher:1959:MT1},
one can construct an invertible change of basis matrix
$T_{\rm G} \in \GL(n)$ which is block-diagonal, but, in general,
not orthogonal; see~\remref{rem:gantmacher} for details.
Second, the numerical approximation of nonlinear matrix functions,
that is currently an important research theme, provides another
possible construction. In particular, our solution approach
for \probref{prob:opt} bears some
resemblance to the work of Higham underlying the computational
approximation of real matrix
square roots of real square matrices via their real Schur form,
see~\cite{Higham:1986:NMSQRT,Higham:1987:CRSQ}
and \cite{Higham:2008:FOM}.\footnote{For a geometric approach
and an account of interesting recent developments in the numerical
approximation of matrix square roots, see~\cite{Sra:2015:MSG}
and references therein.} Our~\exaref{exa:block_and_schur}
and~\remref{rem:real_schur_form} illustrate the relation
between \theref{theo:real_square_roots} and the real Schur
form.

In the next subsection we outline the mechanical background of
our problem. This part may be skipped by readers only interested
in the algebraic development.

\subsection{Optimal Cosserat microrotations and applications in mechanics}
\label{subsec:mechanics}
The term Cosserat theory describes a class of models in nonlinear
solid mechanics incorporating an additional field of rotations.
Such models are also referred to as micropolar
models; see~\cite{Eremeyev:2012:FMM} for an introduction
including extensive references. This type of models dates back
to the original work of the Cosserat brothers~\cite{Cosserat09}
in the early 1900s and was, historically, one of the first
generalized continuum theories.\footnote{The Cosserat brothers
  established the foundations of continuum mechanics with rotational
  degrees of freedom and contributed physically necessary invariance
  requirements for a micropolar continuum theory: the strain energy
  density $W$ in such a theory must be a function of the first
  Cosserat deformation tensor $\mstretch \eqdef \mrot^TF$.
  They never proposed a specific expression for the local strain
  energy density $W = W(\mstretch)$ to model specific materials.}

Let us consider a body $\Omega \subset \R^n$ which is
deformed by a diffeomorphism
$\varphi: \Omega \to \varphi(\Omega) \subset \R^n$
with deformation gradient field
$F \eqdef \nabla\varphi: \Omega \to \GL^+(n)$ and
let us denote the additional field of microrotations by
$\mrot: \Omega \to \SO(n)$. In this context, we introduce
the quadratic Cosserat shear-stretch strain energy density
\begin{equation}
  \wmm(\mrot\,;F) \;\eqdef\;
  \mu\, \hsnorm{\sym(\mrot^TF - \id)}^2
  \,+\,
  \mu_c\,\hsnorm{\skew(\mrot^TF - \id)}^2
\end{equation}
which can be evaluated at every point $x \in \Omega$. The function
$\wmm: \SO(n) \,\times\, \GL^+(n) \to \RPosZ$ depends on
$F = \nabla\varphi$ and $\mrot: \Omega \to \SO(n)$. Note that
in this text, we consider the deformation gradient $F$ as a
parameter, since our interest is to determine energy-minimizing
rotations. The weights $\mu > 0$ and $\mu_c \geq 0$ are given by
the Lam\'e shear modulus $\mu > 0$ from linear elasticity and
the Cosserat couple modulus $\mu_c \geq 0$, see~\cite{Neff_ZAMM05}
for a discussion. The chosen quadratic ansatz for
$W_{\mu,\mu_c}(\mrot\,;F)$ is motivated by a direct extension of
the quadratic energy in the linear theory of Cosserat models,
see, e.g.~\cite{Jeong:2009:NLIC,Neff_Jeong_bounded_stiffness09,
  Neff_Jeong_Conformal_ZAMM08}. It is a contribution to the
variational formulation of \emph{any} geometrically nonlinear,
isotropic, and quadratic Cosserat-micropolar continuum
model, see~\cite{Neff_Cosserat_plasticity05} and~\cite{Cosserat09,Eringen99,Maugin:1998:STPE}.

The polar factor $\polar(F) \in \SO(n)$ is obtained from the right
polar decomposition $F = \polar(F)\,U(F)$ of the deformation gradient
$F \in \GL^+(n)$. It describes the macroscopic rotation of the
continuum. Furthermore, $U(F) \eqdef \sqrt{F^TF} \in \Psym(n)$
describes the stretch and is referred to as the right Biot-stretch
tensor. We note that the singular values $\nu_i$, $i = 1,\ldots,n$,
of the deformation gradient $F \in \GL^+(n)$ are the eigenvalues of
the symmetric positive definite matrix $U \in \Psym(n)$.

It is quite natural to study matrix distance problems in nonlinear
continuum mechanics. Let us illustrate this for the example of the
Euclidean distance function
\begin{equation}
  \dist^2_{\rm Euclid}(F, \SO(n))
  \eqdef
  \min_{R \in \SO(n)}\hsnorm{F - R}^2.
\end{equation}
Conceptually, this function provides a local measure for the
distance of the deformation $\varphi: \Omega \to \varphi(\Omega)$
to the isometric (locally length-preserving) embeddings of the
body $\Omega$ into $\R^n$. The required invariance properties for
isotropy are automatically satisfied. Furthermore, this is consistent
with the requirement that a global isometry of a body
$\Omega \subset \R^n$, i.e., a rigid body motion, does not
produce any deformation energy, because
$F = \nabla (Rx + b) = R \in \SO(n)$
implies
$$
\int\nolimits_{\Omega} \dist^2_{\rm Euclid}(F, \SO(n))\;\dV = 0\;.
$$
Variations on this general theme lead to the study of corresponding
minimization problems on $\SO(n)$ which have been the subject
of multiple contributions, see, e.g.,~\cite{Fischle:2007:PCM,Neff_Osterbrink_hencky13,Fischle:2015:OC2D,Fischle:2015:OC3D,Neff:2014:LMP,Lankeit:2014:MML}.
Note that in classical nonlinear continuum models, the local rotation
of the specimen at a point is not explicitly accounted for
in the strain energy, due to the requirement of frame-indifference.
Thus, in a classical theory, the local rotation of the specimen
induced by a deformation mapping $\varphi$ is always given by the
continuum rotation $\polar(\nabla\varphi)$.

In strong contrast, in Cosserat theory and other generalized continuum
theories (so-called complex materials) with rotational degrees of
freedom, the local rotation $\mrot:\Omega \to \SO(n)$ of the material
appears explicitly. Accordingly, in such a theory, the computation of
locally energy-minimizing rotations provides geometrical insight into
the qualitative mechanical behavior of a particular constitutive model.
The first result in this area apparently dates back to 1940
when Grioli~\cite{Grioli40} proved a remarkable variational
characterization of the polar factor
$\polar(F) \in \SO(3)$
in dimension three. We present a generalization~\cite{Guidugli:1980:EPP}
to arbitrary dimension
\begin{align}
  \label{eq:dist_euclid_abs}
  \argmin{\mrot\,\in\,\SO(n)}{\hsnorm{\mrot^TF - \id}^2}
  &= \left\{ \polar(F) \right\}\;,\quad\text{and}\\
  \min_{\mrot\,\in\,\SO(n)}\hsnorm{\mrot^TF - \id}^2
  &= W_{\rm Biot}(F) \eqdef \hsnorm{U(F) - \id}^2\;.
\end{align}
Grioli's theorem shows that $\polar(F)$ is optimal for
the Cosserat strain energy minimized in \eqref{eq:dist_euclid_abs}.

\begin{rem}[Non-classical rotation patterns in nature]
    \label{rem:non_classical_rotations}
    Grioli's theorem implies that the strain energy minimized
    in \eqref{eq:dist_euclid_abs} can only be expected to
    generate a microrotation field
    $\mrot \approx \polar(\nabla \varphi)$,
    i.e., approximating the classical macroscopic continuum
    rotation.\footnote{Note also that $\mrot = \polar(\nabla \varphi)$
    realizes the classical Biot-energy $\hsnorm{U - \id}^2$.}
    However, non-classical rotation patterns $\mrot$ which deviate
    from $\polar(\nabla \varphi)$ are of interest, since they can
    be observed in many domains of nature. In metals, for example,
    the local rotation of the crystal lattice may differ from the
    continuum rotation considerably which is of importance on the
    meso- and nano-scale. Non-classical counter-rotations of the
    crystal lattice have been observed
    in~\cite{Zaafarani:2006:TITM,Zaafarani:2008:ODRP}
    below nanoindentations in copper single crystals.\footnote{The
      lattice misorientation can be measured by electron backscattered
      diffraction analysis (3D-EBSD).}
    This motivated an analysis of the relaxed polar
    factors in the setting of an idealized nanoindenation in a copper
    single crystal~\cite{Fischle:2016:RPNI}. Another application is in
    geomechanics, where non-classical rotational deformation modes
    may be observed, e.g., in landslides, see~\cite{Muench07_diss}
    and references therein. In the present work, we prove that
    the strain energy $W_{\mu,\mu_c}(R\,;F)$ defined in~\eqref{eq:intro:wmm}
    can produce non-classical microrotations
    $\mrot \approx \rpolar_{\mu,\mu_c}(F)$,
    see also~\cite{Fischle:2015:OC2D,Fischle:2015:OC3D,Fischle:2016:RPNI}.
\end{rem}

The energy minimized in \eqref{eq:dist_euclid_abs}
  is isotropic which allows for a quintessential simplification:
  it allows us to express Grioli's theorem in terms of rotations
  $R \in \SO(n)$ \emph{relative} to the orthogonal polar
  factor $\polar(F)$; see~\cite{Fischle:2015:OC3D} for details.
  In this relative picture, the deformation gradient $F$
  is represented by a diagonal matrix $D \eqdef \diag(d_1,\ldots,d_n)$.
  The entries $d_i = \nu_i > 0$ are the singular values of $F \in \GL^+(n)$.
  Grioli's theorem then takes the following equivalent form
\begin{align}
  \label{eq:dist_euclid_rel}
  \argmin{\mrot\,\in\,\SO(n)}{\hsnorm{RD - \id}^2}
  &= \left\{ \id \right\} \;,\quad\text{and}\\
  \min_{\mrot\,\in\,\SO(n)}\hsnorm{RD - \id}^2
  &= \norm{D - \id}^2\;.
\end{align}
The optimal rotation \emph{relative} to $\polar(F)$ is
given by the identity $\id$. Similarly, exploiting the isotropy of $W_{\mu,\mu_c}(\mrot\,;F)$
in \eqref{eq:intro:wmm}, we obtain the expression
\begin{equation}
  \label{eq:quad}
  \mu\,\hsnorm{\sym(RD - \id)}^2
  \,+\,
  \mu_c\,\hsnorm{\skew(RD - \id)}^2
\end{equation}
in terms of a rotation $R$ relative to $\polar(F)$. For the non-classical
parameter range
$\mu > \mu_c \geq 0$ (see~\cite{Fischle:2015:OC2D} for details),
a non-trivial parameter reduction proved
in~\cite[Lem. 2.2]{Fischle:2015:OC2D} shows that the corresponding
energy-minimizing rotations can be determined by solving~\probref{prob:opt}.

Let us briefly present a highly interesting logarithmic minimization
  problem. To this end, we introduce the logarithmic strain
  energy\footnote{Note that the most general quadratic expression
\begin{equation}
  \label{eq:quad_complete}
  \mu\,\hsnorm{\dev\sym(RD - \id)}^2
  \,+\,
  \mu_c\,\hsnorm{\skew(RD - \id)}^2
  \,+\,
  \,\frac{\kappa}{2}\,\left(\tr{(RD - \id)}\right)^2
\end{equation}
is also of a certain interest, see~\cite{Fischle:2015:OC3D} for a
discussion. The associated optimal rotations are not known to us.}
\begin{equation}
  \label{eq:log}
  W_{\rm log}(R\,;D) \eqdef
  \mu\,\hsnorm{\dev\sym\log(RD)}^2
  \,+\,
  \mu_c\,\hsnorm{\skew\log(RD)}^2
  \,+\,
  \,\frac{\kappa}{2}\,\left(\tr{\log(RD)}\right)^2\;,
\end{equation}
where $\kappa$ denotes the so-called bulk modulus. Technicalities aside,
one can prove that
\begin{align}
  \label{eq:min_log}
  \argmin{R\,\in\,\SO(n)}
         {W_{\rm log}(R\,;D)}
         &= \left\{\id\right\},\quad\text{and}\\
  \min_{R \in \SO(n)} {W_{\rm log}(R\,;D)}
         &= \mu\,\hsnorm{\dev\log D}^2
          + \frac{\kappa}{2}\,\left(\tr{\log D}\right)^2\;,
\end{align}
see~\cite{Neff:2014:LMP,Lankeit:2014:MML,Borisov:2015:SSLI} for
proofs and essential details. Note that these results are closely
related to the geometric observation that certain natural geodesic
distances from $F \in \GL^+(n)$ to the subgroup $\SO(n)$ induce
Hencky-type strain energies~\cite{Neff:2015:GLS}.

The Euclidean distance problem \eqref{eq:dist_euclid_rel}
and the minimization problem for the logarithmic energy \eqref{eq:min_log}
share a remarkable property: the identity $\id \in \SO(n)$ is always
uniquely optimal for any diagonal positive definite $D > 0$. Equivalently,
the polar factor $\polar(F)$ is always the optimal absolute rotation.
In view of~\remref{rem:non_classical_rotations}, we note that the
logarithmic energy produces only classical microrotation patterns, i.e.,
$\mrot \approx \polar(F)$.

In strong contrast, for $\mu > \mu_c \geq 0$, the quadratic Cosserat
shear-stretch energy density \eqref{eq:quad} can produce interesting
non-classical
rotations~\cite{Fischle:2015:OC2D,Fischle:2015:OC3D,Fischle:2016:RPNI},
i.e.,
\begin{equation}
  \argmin{R\,\in\,\SO(n)}
         {\left(\mu\,\hsnorm{\sym(RD - \id)}^2
           \,+\,
           \mu_c\,\hsnorm{\skew(RD - \id)}^2\right)}
         \;\neq\; \{\id\}
\end{equation}
for the optimal relative rotation. Equivalently, it is possible that
\begin{equation}
  \rpolar_{\mu,\mu_c}(F) \;\neq\; \{\polar(F)\}
\end{equation}
which means that the optimal Cosserat microrotations can produce
non-classical rotation patterns, see also~\remref{rem:non_classical_rotations}. This interesting property is an immediate consequence of the
characterization of $\rpolar(D)$ which we prove in the present work.

{\it This paper is structured as follows: after this introduction
  in~\secref{sec:intro}, we proceed
    to~\secref{sec:characterization} which presents a construction
    of an orthonormal basis for any real matrix whose square is
    symmetric such that it takes block-diagonal form with blocks
    of size at most two.
    This block structure allows us to characterize the critical
    points in~\secref{sec:critical_points} for arbitrary
  dimension $n$. This leads to a sequence of decoupled one- and
  two-dimensional subproblems posed, however, on $\O(1)$ and $\O(2)$
  and we continue with the solution of these subproblems
  in~\secref{sec:subproblems}. In~\secref{sec:optimality}
  we extract the globally energy-minimizing optimal Cosserat rotations
  from the set of critical points by a comparison of the realized
  energy levels. It turns out that the optimal rotations and energy
  levels are entirely consistent with previous results for $n = 2,3$.
  We end with a short discussion of the present results
  in~\secref{sec:discussion}.
}
 \countres
\makeatletter{}\section{A block-diagonal representation of real matrices with a real symmetric square}
\label{sec:characterization}

In this section, we present the construction of an orthogonal change
of basis $T \in \O(n)$ for real matrices $X \in \R^{n \times n}$ with
a real symmetric square $S \eqdef X^2 \in \Sym(n)$. The constructed
change of basis preserves the Frobenius matrix norm which allows us to
reduce \probref{prob:opt} to lower-dimensional subproblems.

Let us introduce
\begin{defi}
  We say that $X \in \R^{n \times n}$
  is a real matrix square root of the real symmetric matrix
  $S \in \Sym(n)$, if it solves the quadratic matrix
  equation
  $$X^2 = S \in \Sym(n)\;.$$
\end{defi}
For existence of complex matrix square roots and their
classification, see~\cite{Gantmacher:1959:MT1}
and \cite{Higham:2008:FOM}. The theory of real matrix square
roots is considered in \cite{Higham:1987:CRSQ}.

\begin{exa}
  The identity matrix $\id_2 \in \Sym(2)$ has infinitely many real
  matrix square roots which are simply size two involution matrices.
  They fall into three distinct classes according to their trace:
  \begin{equation}
    X = \id,\quad\quad
    X = -\id,\quad\quad
    X \in \left\{ \begin{pmatrix} a & b\\ c & -a\end{pmatrix},\; a^2 + b\,c = 1\right\}\;.
  \end{equation}
\end{exa}

\begin{exa}
  The negative identity matrix $-\id_2 \in \Sym(2)$ has the real matrix
  square roots
  \begin{equation}
    X \in \left\{ \begin{pmatrix} a & b\\ c & -a\end{pmatrix},\; a^2 + b\,c = -1\right\}\;.
  \end{equation}
\end{exa}

\begin{exa}
  A negative identity matrix of odd size $-\id \in \Sym(2k-1)$
  does not have real matrix square roots, since its determinant
  is negative.
\end{exa}

We now provide a simple criterion for $X \in \R^{2 \times 2}$
to be a square root of \emph{some} symmetric real
matrix, i.e., $X^2 = (X^T)^2$.
This will be useful in~\secref{sec:subproblems}.
\begin{lem}
  \label{lem:root_2x2}
  A real matrix $X \in \R^{2 \times 2}$ is a real matrix square root
  of a real symmetric matrix $S = X^2 \in \Sym(2)$ if and only if
  $X \in \Sym(2)$ or $\tr{X} = 0$.
\end{lem}
\begin{proof}
  The Cayley-Hamilton theorem implies that
  $$
  S = X^2 = \tr{X}\,X - (\det X)\,\id\;.
  $$
  Since the square $S = X^2$ is symmetric, we have
  $$
  0 = \skew(S) = \skew(\tr{X}\,X - (\det X)\,\id) = \tr{X}\,\skew(X)\;.
  $$
  This finishes the argument.
\end{proof}

In what follows, we write $E_{\lambda_i} \subseteq \R^n$ to denote the
maximal real eigenspace of a real symmetric matrix $S \in \Sym(n)$
associated to a given eigenvalue $\lambda_i$, $i = 1,\ldots,m$.

\begin{lem}[Eigenspaces of $S = X^2 \in \Sym(n)$ are preserved by $X$]
  \label{lem:no_mix}
  Let $X \in \R^{n \times n}$ with symmetric square
  $S \eqdef X^2 \in \Sym(n)$ and let $\lambda \in \R$ be an
  eigenvalue of $S$. Then $X$ preserves the eigenspace
  $E_{\lambda}$ of its square $S$, i.e.,
  $$
  X E_{\lambda} \subseteq E_{\lambda}\;.
  $$
\end{lem}
\begin{proof}
  The operators $X$ and $X^2$ commute.
\end{proof}

\begin{cor}
  \label{cor:square_root_eigenblocks}
  Let $X \in \R^{n \times n}$ be a real matrix
  square root of $S = X^2 \in \Sym(n)$ and let $T \in \O(n)$ such that
  $$
  \widetilde{S} \eqdef T^{-1}ST
  = \diag(\underbrace{\lambda_1,\ldots,\lambda_1}_{\dim E_{\lambda_1}},
  \underbrace{\lambda_2,\ldots,\lambda_2}_{\dim E_{\lambda_2}},
  \ldots,\underbrace{\lambda_m,\ldots,\lambda_m}_{\dim E_{\lambda_m}})\;.
  $$
  Then the transformed matrix
  \begin{equation}
    \widetilde{X} \eqdef T^{-1}XT = \diag(\widetilde{X}_{1},\widetilde{X}_{2},\ldots,\widetilde{X}_{m})
  \end{equation}
  is block-diagonal with square blocks $\widetilde{X}_i$ of size
  $\dim E_{\lambda_i}$, $i = 1,\ldots,m$, that satisfy
  $$
  \widetilde{X}_i^2 = \widetilde{S}_i = \lambda_i \id\;.
  $$
\end{cor}

\begin{rem}
  The preceding \coref{cor:square_root_eigenblocks} reduces
  the subsequent characterization of real matrix square roots of
  symmetric matrices significantly, because it shows
  that it suffices to consider each of the $X$-invariant eigenspaces
  $E_{\lambda_i}$, $i = 1,\ldots,m$, of $S$
  \emph{individually}.
\end{rem}

  Our next step is to construct a suitable \emph{orthonormal} basis
  for each individual eigenspace $E_{\lambda_i}$, $i = 1,\ldots,m$, of
  $S \in \Sym(n)$. As we shall see in the following, this yields a
  change of basis with transition matrix $T \in \O(n)$ which is
  \emph{adapted} to the structure of \probref{prob:opt}. In particular,
  it preserves the Frobenius matrix norm.

  Before we proceed with our proposed construction, we want to briefly
  describe an intuitive and \emph{apparently} similar
  (but completely different) approach relying on the classification
  of the set of all complex matrix square roots due to
  Gantmacher~\cite{Gantmacher:1959:MT1}.
  \begin{rem}[Lack of orthogonality in Gantmacher's representation]
    \label{rem:gantmacher}
    To make our point, it suffices to consider Gantmacher's
    classification of the complex matrix square roots for
    non-degenerate $A \in \GL(n,\C)$. We follow the exposition
    of Higham~\cite[Thm. 1.24]{Higham:2008:FOM} including the
    notation. In our setting $A = \widetilde{S}$ is real diagonal,
    i.e., it is in Jordan canonical form with Jordan blocks
    of size one. Due to Gantmacher's classification, all
    \emph{complex} matrix square roots $X \in \C^{n \times n}$
    which satisfy $X^2 = A = \widetilde{S}$ are of the form    \footnote{Here, we use the multi-valued convention to
      denote the set of all complex square roots by the
      symbol $\sqrt{\cdot}$.}
  \begin{equation}
    X = U\diag(\sqrt{\lambda_1},\ldots,\sqrt{\lambda_m})U^{-1}\;,
  \end{equation}
  where $U \in \GL(n,\C)$ is arbitrary, but required to commute
  with the Jordan matrix $A = \widetilde{S}$, i.e.,
  \begin{equation}
    U\diag(\lambda_1,\ldots,\lambda_m) = \diag(\lambda_1,\ldots,\lambda_m)U\;.
  \end{equation}
  Consider now a \emph{real} matrix square root $X \in \R^{n \times n}$.
  It is not hard to see that the  real Jordan form of $X$ is obtained
  by an invertible change of basis matrix
  $T_{\rm G} \in \GL(n,\R)$ with the property that
  $\widetilde{X} = T_{\rm G}^{-1}X T_{\rm G} \in \R^{n\times n}$
  is block-diagonal with blocks of size one and two. Here, the
  blocks of size one correspond to positive eigenvalues of
  $A = \widetilde{S}$ and the blocks of size two correspond to
  negative eigenvalues of $A = \widetilde{S}$. Unfortunately,
  $T_{\rm G} \in \GL(n,\R)$ is in general not orthogonal
  and the change of basis does not preserve the Frobenius
  matrix norm in~\probref{prob:opt}.
  \end{rem}
Let us now proceed with the construction of a suitable \emph{orthogonal}
change of basis matrix $T \in \O(n)$ on which our subsequent analysis
of \probref{prob:opt} is based.

We briefly recall the definition of the orthogonal complement
$V^\perp$ of a linear subspace $V \subseteq \R^n$,
  $$
  V^\perp \;\eqdef\; \setdef{w \in \R^n}{w \perp V}
  = \setdef{w \in \R^n}{\forall v \in V: \scalprod{v}{w} = 0}\;,
  $$
  which induces an orthogonal decomposition of
  $\R^n = V \oplus_\perp V^\perp$. In what follows, we
  exploit the well-known fact that for $Y \in \R^{n \times n}$,
  \begin{equation}
    \label{eq:inv_vperp}
    Y V^\perp \subseteq V^\perp \quad \isequivto \quad Y^TV \subseteq V\;.
  \end{equation}
  Indeed, let $w \in V^\perp$, then
  $0 = \scalprod{Y w}{v} = \scalprod{w}{Y^Tv}$.
  Since the choice of $w \in V^\perp$ was arbitrary, we have that
  $Y^Tv \perp V^\perp$ which shows $Y^Tv \in V$,
  because $\R^n = V \oplus_\perp V^\perp$. The reverse implication
  is completely analogous.

\begin{lem}[Block lemma]
  \label{lem:real_square_roots}
  Let $Y \in \R^{n \times n}$ with square $S \eqdef Y^2 = \lambda \id$, $\lambda \in \R$.
  Then there exists an orthogonal transformation $T \in \O(n)$ such that
  the matrix
  \begin{align}
    \widetilde{Y} \eqdef T^{-1}YT
    = \diag(\widetilde{Y}_1,\widetilde{Y}_2,\ldots,\widetilde{Y}_r)
    = \begin{pmatrix}
    \widetilde{Y}_1  & 0               & \ldots & 0\\
    0                & \widetilde{Y}_2 & \ddots & \vdots\\
    \vdots           & \ddots          & \ddots & 0\\
    0                & \ldots          & 0      & \widetilde{Y}_r
  \end{pmatrix}
  \end{align}
  is block-diagonal, with blocks $\widetilde{Y}_i$, $i = 1,\ldots,r$,
  of size one or two satisfying $\widetilde{Y}_i^2 = \lambda \id$.
\end{lem}

\begin{proof}
  The proof proceeds by induction on $n$. The base case of induction
  $n \in \{1,2\}$ holds, since $Y$ is already block-diagonal
  with blocks of size one or two. For the induction step let us assume
  that the statement holds for matrices of size $n-1$ and $n-2$.
  \\
  \smallskip
  Our strategy is to prove the existence of a one- or two-dimensional
  subspace $V$ of $\R^n$ such that both $V$ and its orthogonal
  complement $V^\perp$ are left invariant by $Y$, i.e.,
\begin{equation}
  \dim V \in \{1,2\},\quad Y V \subseteq V
  \quad\text{and}\quad
  Y V^\perp \subseteq V^\perp\;.
\end{equation}
Thus, if we pick an orthonormal basis of $V$ and $V^\perp$,
this is equivalent to the statement that orthogonal conjugates
of $Y$ and $Y^T$ are block matrices of the form
\begin{equation}
  \label{eq:blockmat_induction}
  Q^{-1}YQ =
  \left(\begin{array}{c|c}
    \widetilde{Y}_1 & 0\\
    \hline
    0               & Z
  \end{array}\right)\;.
\end{equation}
Equivalently, $V$ is invariant under both $Y$ and $Y^T$.
Since $(Q^{-1}YQ)^2 = Q^{-1}Y^2Q = \lambda\id$, we get
$Z^2 = \lambda \id$, so by the induction assumption there exists
$T_0$ such that
\begin{equation}
  T_0^{-1}ZT_0 =
  \begin{pmatrix}
    \widetilde{Z}_1  & 0               & \ldots & 0\\
    0                & \widetilde{Z}_2 & \ddots & \vdots\\
    \vdots           & \ddots          & \ddots & 0\\
    0                & \ldots          & 0      & \widetilde{Z}_s
  \end{pmatrix}\;.
\end{equation}
Then the orthogonal matrix
\begin{equation}
  T = Q\begin{pmatrix}\id & 0\\ 0 & T_0\end{pmatrix} \in \O(n)
\end{equation}
satisfies
\begin{equation}
  T^{-1}YT =
  \begin{pmatrix}
    \widetilde{Y}_1  & 0               & \ldots & 0\\
    0                & \widetilde{Z}_1 & \ddots & \vdots\\
    \vdots           & \ddots          & \ddots & 0\\
    0                & \ldots          & 0      & \widetilde{Z}_s
  \end{pmatrix}\;,
\end{equation}
which completes the induction step.
\smallskip

To finish the argument, we have to construct a $Y$- and $Y^T$-invariant
subspace $V$ of $\R^n$ of dimension one or two.\\
\smallskip

Since $(Y^T)^2 = S^T = S = \lambda\id$, the symmetric matrices
$YY^T$ and $Y^TY$ commute
\begin{equation}
  (YY^T)(Y^TY) = Y(Y^T)^2Y = Y (\lambda\id) Y = \lambda S = \lambda^2\id = (Y^TY)(YY^T).
\end{equation}
Therefore, the operators $YY^T$ and $Y^TY$ are simultaneously
diagonalizable and we can find a common eigenvector $w$ of both.
Let us normalize $w$ so that $\norm{w} = 1$ and note that there
exist values $\alpha, \beta \in [0,\infty)$ satisfying
\begin{equation}
  Y^TYw \;=\; \alpha w
  \quad\quad \text{and} \quad\quad
  YY^Tw \;=\; \beta w\;.
\end{equation}
Our next step is to choose the invariant subspace $V$. We have to
distinguish several cases.\\[2em]
{\bf Case 1:} $Yw \in \vspan{w}, Y^Tw \in \vspan{w}$, in other words,
  $w$ is an eigenvector of $Y$ and $Y^T$. We select $V = \vspan{w}$
  and construct an orthogonal matrix with first column given by
  $q_1 = w$, i.e.,
  \begin{equation}
    Q = (w | q_2 | \ldots | q_n) \in \O(n)\;.
  \end{equation}
  An associated change of basis for $Y$ and $Y^T$ introduces
  the following zero patterns
  \begin{equation}
    Q^{-1}YQ =
    \begin{pmatrix}
      *     & \vline && * &   &\\
      \hline
      0     & \vline &&   &   &\\
      \cdot & \vline &&   &   &\\
      \cdot & \vline && * &   &\\
      \cdot & \vline &&   &   &\\
      0     & \vline &&   &   &
    \end{pmatrix}
    \quad\text{and}\quad
    Q^{-1}Y^TQ =
    \begin{pmatrix}
      *     & \vline && * &   &\\
      \hline
      0     & \vline &&   &   &\\
      \cdot & \vline &&   &   &\\
      \cdot & \vline && * &   &\\
      \cdot & \vline &&   &   &\\
      0     & \vline &&   &   &
    \end{pmatrix}\;.
  \end{equation}
  Since $Q^{-1} Y^T Q = (Q^{-1} Y Q)^T$ these matrices are transposes
  of each other which implies that we obtain a block matrix of
  the form
  \begin{equation}
    Q^{-1} Y Q =
    \begin{pmatrix}
      *     & \vline & 0 & \cdot & \cdot & 0\\
      \hline
      0     & \vline &   &   &   &\\
      \cdot & \vline &   &   &   &\\
      \cdot & \vline &   & * &   &\\
      \cdot & \vline &   &   &   &\\
      0     & \vline &   &   &   &
    \end{pmatrix}\;,
  \end{equation}
  which is of the form described in \eqref{eq:blockmat_induction}.\\[2em]

  {\bf Case 2:} $Yw \in \vspan{w}, Y^Tw \notin \vspan{w}$, in other words $w$ is an
  eigenvector of $Y$ but not of $Y^T$. Consider the subspace
  $V = \vspan{w, Y^Tw}$. Then the image of $V$ under $Y$ satisfies
  \begin{align}
    YV   &= \vspan{Yw,YY^Tw} \subseteq \vspan{w,w} \subseteq V\\
    Y^TV &= \vspan{Y^Tw, (Y^T)^2w} \subseteq \vspan{Y^Tw,\lambda w} \subseteq V\;.
  \end{align}
  We now pick an orthonormal basis $w_1,w_2$ of
  $V = \vspan{w_1,w_2} = \vspan{w, Y^Tw}$ and extend it to an orthogonal
  matrix
  \begin{equation}
    Q = (w_1| w_2 | q_3 | \ldots | q_n) \in \O(n)\;.
  \end{equation}
  Then, similar to Case 1, an associated change of basis for $Y$ and
  $Y^T$ introduces a zero pattern
  \begin{equation}
    Q^{-1}YQ =
    \begin{pmatrix}
      *     & *     & \vline &        & * &   &\\
      *     & *     & \vline &        & * &   &\\
      \hline
      0     & 0     & \vline &        &   &   &\\
      \cdot & \cdot & \vline &        &   &   &\\
      \cdot & \cdot & \vline &        & * &   &\\
      \cdot & \cdot & \vline &        &   &   &\\
      0     & 0     & \vline &        &   &   &
    \end{pmatrix}
    \quad\text{and}\quad
    Q^{-1}Y^TQ =
    \begin{pmatrix}
      *     & *     & \vline &        & * &   &\\
      *     & *     & \vline &        & * &   &\\
      \hline
      0     & 0     & \vline &        &   &   &\\
      \cdot & \cdot & \vline &        &   &   &\\
      \cdot & \cdot & \vline &        & * &   &\\
      \cdot & \cdot & \vline &        &   &   &\\
      0     & 0     & \vline &        &   &   &
    \end{pmatrix}\;.
  \end{equation}
  As before, since $Q^{-1} Y^T Q = (Q^{-1} Y Q)^T$ the two matrices
  are transposes of each other which creates a $2$-block in the
  upper left corner
  \begin{equation}
    Q^{-1} Y Q =
    \begin{pmatrix}
      *     & *     & \vline & 0 & \cdots & 0  &\\
      *     & *     & \vline & 0 & \cdots & 0  &\\
      \hline
      0     & 0     & \vline &   &   &   &\\
      \cdot & \cdot & \vline &   &   &   &\\
      \cdot & \cdot & \vline &   & * &   &\\
      \cdot & \cdot & \vline &   &   &   &\\
      0     & 0     & \vline &   &   &   &
    \end{pmatrix}\;,
  \end{equation}
  which is of the form described in \eqref{eq:blockmat_induction}.\\[2em]
  {\bf Case 3:} $Yw \notin \vspan{w}$, in other words $w$ is \emph{not} an
  eigenvector of $Y$. We consider the subspace $V = \vspan{w, Yw}$. The
  inclusion
  \begin{equation}
    YV = \vspan{Yw, Y^2w} = \vspan{Yw, \lambda w} \subseteq V
  \end{equation}
  is immediate. In order to prove the invariance $Y^TV \subseteq V$, we
  need to consider the following two subcases:\\[1em]
  {\bf Case 3a:} $\lambda \neq 0$. In this case $Y$ and $Y^T$ are invertible
  and so $Y^TYw = \alpha w$ with $\alpha > 0$. This allows us to express
  $w$ as follows
  \begin{equation}
    \left(\frac{1}{\alpha} Y^TY\right)w = \frac{\alpha}{\alpha}\, w = w\;.
  \end{equation}
  We have to compute
  \begin{equation}
    Y^T V = \vspan{Y^Tw,Y^TYw} = \vspan{Y^Tw,\alpha w}\;.
  \end{equation}
  To this end, we expand
  \begin{equation}
  Y^Tw = Y^T\left(\frac{1}{\alpha}Y^T Y\right)w
       = \frac{1}{\alpha}\,(Y^2)^T Yw
       = \frac{1}{\alpha}\, S^TYw
       = \frac{1}{\alpha}\, Y^3w
       = \frac{1}{\alpha}\, YSw
       = \frac{\lambda}{\alpha}\, Yw \in V
  \end{equation}
  which shows that $Y^TV \subseteq V$.\\[2em]
  {\bf Case 3b:} $\lambda = 0$. Consider the product
  \begin{equation}
  \left(Y^TY\right)\left(YY^T\right)w
  = Y^TSY^Tw
  = S^2w
  = \lambda^2 w
  = 0\;.
  \end{equation}
  Since we also have, $Y^TYw = \alpha w$ and $YY^Tw = \beta w$,
  it follows that
  $\left(Y^TY\right)\left(YY^T\right)w = \alpha\beta w = 0$.
  Hence, $\alpha \beta = 0$. If $\beta = 0$, then
  \begin{equation}
  YY^Tw = 0 \quad\Longrightarrow\quad \scalprod{w}{YY^Tw} = 0
  \quad\Longrightarrow\quad \scalprod{Y^Tw}{Y^Tw} = \norm{Y^Tw}^2 = 0\;.
  \end{equation}
  Since $Y^T w = 0 \in V$, the subspace $V$ is invariant under both
  $Y$ and $Y^T$. The second case $\alpha = 0$ is not possible. To see
  this, we similarly compute
  \begin{equation}
    Y^TYw = 0 \quad\Longrightarrow\quad Yw = 0
  \end{equation}
  which shows that $w$ is an eigenvector of $Y$. This contradicts
  our assumptions for Case 3 (but note that this situation is handled
  in Case $1$ or $2$).\\[1em]
  This completes the construction of the invariant subspace $V$
  and the proof of the lemma.
\end{proof}

\begin{rem}
  The case $\lambda > 0$ of \leref{lem:real_square_roots} can be
  deduced from the theory of principal angles
  (see, e.g,~\cite{Galantai:2013:PPM}) for the eigenspaces of
  $Y$ with eigenvalues $\sqrt{\lambda}$ and $-\sqrt{\lambda}$.
  We are not aware of a similar connection in the case $\lambda \leq 0$.
\end{rem}

\begin{rem}
The condition on $Y$ is also sufficient, i.e., any matrix with the
described block structure is a real matrix square root of a symmetric
matrix. It is also possible to show that solutions to
$Y^2 = \lambda \id$ exist if and only if $\lambda \geq 0$, or $n$ is
even.
\end{rem}

We are now ready to formulate the main result of this section.
\begin{theo}
  \label{theo:real_square_roots}
  For any real square matrix $X \in \R^{n \times n}$ with symmetric
  square $X^2 \in \Sym(n)$ there exists an orthogonal change
  of coordinates $T \in \O(n)$ such that the transformed
    matrix
  \begin{align}
    \widetilde{X} \eqdef T^{-1}XT
    = \diag(\widetilde{X}_1,\widetilde{X}_2,\ldots,\widetilde{X}_r)
    = \begin{pmatrix}
    \widetilde{X}_1  & 0               & \ldots & 0\\
    0                & \widetilde{X}_2 & \ddots & \vdots\\
    \vdots           & \ddots          & \ddots & 0\\
    0                & \ldots          & 0      & \widetilde{X}_r
  \end{pmatrix}
  \end{align}
  is block-diagonal with square blocks $\widetilde{X}_j, j = 1,\ldots,r$,
  that are either of size one or two. Each block $\widetilde{X}_j$ is a
  real square root of a multiple of an identity matrix $\id$, i.e.,
  $$\widetilde{X}_j^2 = \mu_j \id,\quad \mu_j \in \R\;.$$
\end{theo}
\begin{proof}
  It suffices to apply \leref{lem:real_square_roots} to each eigenblock
  of $S = X^2$.
\end{proof}

\begin{rem}
  Each eigenspace $E_{\lambda_i}$ of $S$ in \leref{lem:no_mix} is
  possibly decomposed into
  multiple subspaces by \leref{lem:real_square_roots}. As a
  result, the eigenvalues $\mu_j$, $1 \leq j \leq r$,
  of $\widetilde{X}_j^2$
  in \theref{theo:real_square_roots} are equal to the eigenvalues
  $\lambda_i$, $1 \leq i \leq m$, in the notation of
  \leref{lem:no_mix} with, possibly,
  different indices. Several $\mu_j$ in the statement of
  \theref{theo:real_square_roots}
  may be equal to the same $\lambda_i$ in the sense of
  \leref{lem:no_mix}. For example, we might have the following
  $$
  (\mu_1,\mu_2,\mu_3,\mu_4,\mu_5)
  = (\lambda_1,\lambda_1,\lambda_2,\lambda_3,\lambda_3)\;.
  $$
\end{rem}

\begin{rem}
  \label{rem:real_square_roots}
  An equivalent reformulation of the theorem is the following.
  For a matrix $X$ whose square is symmetric, there exists a decomposition
  of $\R^n$ into an orthogonal direct sum of $X$-invariant subspaces
  $V_i$ of dimension one or two such that $X^2$ is a multiple of
  the identity matrix on each $V_i$. The list of columns of the change
  of basis matrix $T \in \O(n)$ in \theref{theo:real_square_roots}
  is obtained by concatenation of orthonormal bases of $V_i$. Note
  that each $V_i$ is also invariant under $X^T$.
\end{rem}

\begin{rem}
  Given $X$ and $S$ the decomposition into invariant subspaces is not
  unique. In particular, a subspace of dimension two can sometimes
  be further decomposed into two one-dimensional subspaces.
\end{rem}

\begin{rem}
  Our description of real matrices which square to a real symmetric
  matrix resembles the well-known characterization
  of the group of real orthogonal matrices $\O(n)$. Every
  orthogonal matrix is orthogonally conjugated to a block
  diagonal matrix with blocks of size one and two,
  see, e.g.,~\cite[Thm. 12.5,p. 354]{Gallier:2011:GMA}.
\end{rem}

The following example illustrates the block-diagonal form
stated in~\theref{theo:real_square_roots}. At the same
time it illuminates the relation of the construction to
the real Schur form of a real matrix square root
described in~\cite{Higham:1987:CRSQ}.
\begin{exa}
  \label{exa:block_and_schur}
  We use the notation of the Block~\leref{lem:real_square_roots}
  and consider the matrix
  \begin{equation}
Y=
\begin{pmatrix*}[r]
  1 & 0 & 1 & 1\\
  0 & 1 & 1 & 1\\
  0 & 0 & -1& 0\\
  0 & 0 & 0 & -1
\end{pmatrix*} \in \R^{4 \times 4}\;.
\end{equation}
Note that $Y^2 = \id$. Clearly, $Y$ is not block-diagonal, but it is
in real Schur form~\cite[Thm. 6]{Higham:1987:CRSQ}.
In this representation, the value of $\hsnorm{Y}^2$ cannot be decomposed
into diagonal contributions. As we noted before, the orthogonal
transformation $T \in \O(n)$ described by \leref{lem:real_square_roots}
is not unique. Let us choose
\begin{equation}
T=
\frac{1}{\sqrt 2}
\begin{pmatrix*}[r]
-1 & 0 & 0 & 1\\
 1 & 0 & 0 & 1\\
 0 &-1 & 1 & 0\\
 0 & 1 & 1 & 0
\end{pmatrix*}
\quad\text{which yields}\quad
\widetilde{Y} \eqdef T^{-1}YT =
\begin{pmatrix*}[r]
1 & 0 & 0 & 0\\
0 & -1 & 0 & 0\\
0 & 0 & 1& 0\\
0 & 0 & 2 & -1
\end{pmatrix*}\;.
\end{equation}
Hence, in the orthonormal basis defined by $T$, we obtain a block-diagonal
representation $\widetilde{Y}$ with blocks of size at most two and
satisying $\widetilde{Y}^2 = \id$. The Frobenius norm of $Y$ is now
composed of diagonal contributions
\begin{equation}
  8 = \hsnorm{Y}^2
    = \hsnorm{T^{-1}YT}^2
    = \hsnorm{\widetilde{Y}}^2
    = \hsnorm{\begin{pmatrix*}[r] 1 & 0\\ 0 & -1\end{pmatrix*}}^2 +
      \hsnorm{\begin{pmatrix*}[r] 1 & 2\\ 0 & -1\end{pmatrix*}}^2
  = 2 + 6\;.
\end{equation}
In the present example, the transformed matrix $\widetilde{Y}$ is
not in real Schur form, since the eigenvalues of the lower right
$2\times 2$-block are not complex conjugates. Note that, if we
flip the third and fourth columns of $T$, then $\widetilde{Y}$
is, again, in real Schur form.
\end{exa}

\begin{rem}
  \label{rem:real_schur_form}
  The above example shows that the block-diagonal form guaranteed
  by~\theref{theo:real_square_roots} is not in general a real Schur
  form. It is, however, always possible to find an orthogonal change
  of basis that yields the block-diagonal representation of
  \theref{theo:real_square_roots} in real Schur form, a fact
  which we will not use in the paper. Indeed, all blocks of
  size one and all blocks of size two with $\lambda < 0$ are already
  in real Schur form. For a size two block with $\lambda \geq 0$,
  we can pick an orthonormal basis that begins with an
  eigenvector.
\end{rem}
 \countres
\makeatletter{}\section{Critical points of the Cosserat shear-stretch energy}
\label{sec:critical_points}
In this section we investigate the critical points $R \in \SO(n)$
of the objective function
\begin{equation}
  W(R\,;D)\;=\;\hsnorm{\sym(RD - \id)}^2
\end{equation}
in~\probref{prob:opt}. Since the objective function $W(R\,;D)$ is
polynomial, we can proceed by taking derivatives along curves in
the matrix group $\SO(n)$. Regarding the diagonal parameter
$D = \diag(d_1,\ldots,d_n)$,
we make a mild \assref{ass:D}, in particular $d_i \neq 0$,
and give a complete description of the critical points
in that case. For some of our conclusions, we have to
make the more restrictive \assref{ass:Dpos} enforcing,
in particular, positive $d_i > 0$.

In \leref{lem:symmetric_square_condition} we show that the matrix
$X(R) \eqdef RD - \id$ satisfies a symmetric square condition
precisely when $R \in \SO(n)$ is a critical
point of the objective function $W(R\,;D)$. Hence $X(R)$ is
\emph{orthogonally} similar to a block-diagonal matrix
by \theref{theo:real_square_roots}
at the critical points. This is exploited in the \emph{key result}
\leref{lem:sim_inv} of this section, which shows that the
block-structure of $X(R) = RD - \id$ at the critical points is
simultaneously inherited by \emph{both} $R$ and $D$. This observation
then allows to break \probref{prob:opt} down into decoupled
subproblems of dimension one and two.

Since $\SO(n) \subset \R^{n \times n}$ is a Lie matrix group, we can
identify the Lie algebra $\so(n)$ with the tangent space
$T_\id\SO(n) \subset \R^{n \times n}$ of $\SO(n)$ at the identity.
It is well-known that the tangent space $\so(n)$ is given by the
skew-symmetric matrices $\Skew(n)$, see, e.g.,~\cite{Baker:2012:MG}.
Furthermore, the Frobenius inner product $\scalprod{X}{Y} \eqdef \tr{X^TY}$
gives rise to an orthogonal decomposition of the vector space
of \emph{all} real square matrices
$$
\R^{n \times n} = \Sym(n) \oplus_{\perp} \Skew(n)\;.
$$
Our analysis of the critial points is based on the following algebraic
stationarity condition obtained from the Euler-Lagrange equations.

\begin{lem}[Symmetric square condition]
  \label{lem:symmetric_square_condition}
  Let $D \eqdef \diag(d_1, \ldots, d_n) \in \R^{n \times n}$ be a diagonal
  matrix and let $X(R) \eqdef RD - \id$. Then a rotation $R \in \SO(n)$
  is a critical point of the objective function
  $$
  W(R\,;D) \;\eqdef\; \hsnorm{\sym(RD - \id)}^2
  $$
  if and only if $X(R)$ satisfies the symmetric square condition
  $$
  X(R)^2 \eqdef (RD - \id)^2 \in \Sym(n)\;.
  $$
\end{lem}
\begin{proof}
In order to compute critical points in the submanifold
$\SO(n) \subset \R^{n \times n}$, we have to locate zeroes of
the tangent mapping
${\rm d}W: T\SO(n) \to T\R^{n\times n} \eqiso \R^{n \times n}$.
To this end, we compute the derivatives of the energy
$W(R\,;D)$ along a family of smooth curves
\begin{equation}
  c_A: (-\epsilon, \epsilon) \to \SO(n),\quad c_A(t) \eqdef \exp(tA)R \in \SO(n),\quad A \in \so(n),
\end{equation}
in the manifold of rotations. The right-trivialization
of the tangent space at $R \in \SO(n)$ allows to identify
$T_R \SO(n) = \so(n) \cdot R = \Skew(n) \cdot R$
and so we can always express a tangent vector $\xi \in T_R \SO(n)$
in the form $\xi = AR \in \Skew(n) \cdot R$. This family of
curves satisfies
\begin{equation}
  \forall\, \xi = AR \in T_R \SO(n):\quad
  \left.\frac{\rm d}{\rm dt}\right|_{t = 0}\; c_A(t) = AR = \xi\;.
\end{equation}
Thus, for every possible tangent direction $\xi = AR \in T_R\SO(n)$,
there is precisely one curve of the family which emanates from
$R \in \SO(n)$ into this direction $\xi$.

A rotation $R$ is a critical point of the energy $W(R\,;D)$ if and
only if
$$
\forall A \in \so(n):\quad \frac{\rm d}{\rm dt}\left. (W \circ c_A)(t) \right\vert_{t=0} \;=\; 0\;.
$$

It is well-known that the matrix exponential is given by $(\id + tA)$
to first order in $t$ and we write $\exp(tA) \sim (\id + tA)$. Thus,
by the chain rule, we also have
$$
(W\circ c_A)(t) \sim (W\circ (\id + tA)R)(t)\;.
$$
We expand the expression
\begin{align*}
  W \circ (\id + tA)R
  &= \hsnorm{\sym((1+tA)RD - \id)}^2
  = \hsnorm{\sym(RD - \id) + t\sym(ARD)}^2\\
  &= \hsnorm{\sym(RD - \id)}^2
  + 2t\, \scalprod{\sym(RD - \id)}{\sym(ARD)}
  + t^2\, \hsnorm{\sym(ARD)}^2
\end{align*}
and obtain the expression for the first derivative ${\rm d}W$ from
the term linear in $t$. In other words
\begin{equation}
  \frac{\rm d}{\rm dt}\left. (W \circ c_A)(t) \right\vert_{t=0} = 2\,\scalprod{\sym(RD - \id)}{\sym(ARD)}\;.
\end{equation}
Hence, a point $R$ is a critical point for the energy $W$ if and only
if it satisfies
$$
\forall A \in \so(n):\quad \sym(RD - \id) \perp \sym(ARD)\;.
$$
Since $\Sym(n) \perp \Skew(n)$, we may add $\skew(ARD)$ on
the right hand side which gives us the equivalent condition
$$
\forall A \in \so(n):\quad \sym(RD - \id) \perp ARD\;.
$$
Expanding the definition of the Frobenius inner product,
we find
\begin{align}
  0 = \scalprod{\sym(RD - \id)}{ARD}
  &= \tr{\sym(RD - \id)^T ARD}
  = \tr{RD\sym(RD - \id) A}\notag\\
  &= \scalprod{\sym(RD - \id)(RD)^T}{A}\;.
\end{align}
Since this condition must hold for all $A \in \Skew(n)$,
it follows that
$$
\sym(RD - \id)DR^T \in \Sym(n)\;.
$$
We now multiply by a factor of $2$ and expand the definition of
$\sym(X) \eqdef \frac{1}{2}(X + X^T)$ which leads us to
\begin{align}
2 \sym(RD - \id)DR^T
&= (RD + DR^T - 2\id)DR^T
 =  RD^2R^T + (DR^T)^2 - 2DR^T\notag\\
&= (DR^T - \id)^2 + (RD^2R^T - \id).
\end{align}
The second term on the right hand side is always symmetric
and the effective condition for a critical point is thus
\begin{equation}
  (DR^T - \id)^2 \in \Sym(n)\;.
\end{equation}
Finally, observing that symmetry is invariant under
transposition, we conclude that
\begin{equation}
  \label{eq:sym_sq_cond_deriv}
  \left((DR^T - \id)^2\right)^T = (RD - \id)^2 \in \Sym(n)
\end{equation}
is a sufficient and necessary condition for a critical point
$R \in \SO(n)$ of $W(R\,;D)$.
\end{proof}

\begin{rem}
  We immediately observe that $R = \id$ solves the condition
  \eqref{eq:sym_sq_cond_deriv} and is always a critical point
  of the energy $W(R\,;D) \eqdef \hsnorm{\sym(RD - \id)}^2$.
  However, in general, it will not be the global minimizer.
\end{rem}

Our next step is to apply~\theref{theo:real_square_roots}
and~\remref{rem:real_square_roots} to the special
case $X(R) = RD - \id$. As we shall see, this implies quite restrictive
conditions on $R \in \SO(n)$.

Let us make the following assumption on the diagonal matrix $D$.
\begin{ass}
  \label{ass:D}
  The entries of the diagonal matrix $D = \diag(d_1,\ldots,d_n)$, which
  parametrizes the energy $W(R\,;D)$, do not vanish and do not cancel
  each other additively, i.e.,
  $$
  d_i \neq 0 \quad\text{and}\quad d_i + d_j \neq 0,\quad 1\leq i,j \leq n\;.
  $$
\end{ass}
This ensures that $\ker(D) = {\bf 0}$ and that any $D^2$-invariant
subspace is also $D$-invariant. Note that if the entries of
$D = \diag(d_1,\ldots,d_n)$ are positive, this assumption is
satisfied. For the original problem in Cosserat theory which stimulated
the present
work~\cite{Fischle:2015:OC2D,Fischle:2015:OC3D,Fischle:2016:RPNI},
the entries of $D$ are the singular values
$\nu_i > 0$, $i = 1,\ldots,n$,
of the deformation gradient $F \in \GL^+(n)$.

The following insight is a key to our discussion.
\begin{lem}[Simultaneous invariance of $R$ and $D$]
  \label{lem:sim_inv}
  Suppose that the eigenvalues of $D$ satisfy the above assumption.
  Let $V$ be a subspace invariant under $X(R) = RD - \id$, such that
  $V^\perp$ is also invariant under $X(R)$. Then both $V$ and
  $V^\perp$ are invariant under $D$ and $R$.
\end{lem}

\begin{proof}
Recall first that $V$ and $V^\perp$ are both invariant under $X(R)$
if and only if $V$ is invariant under both $X(R)$ and
$X(R)^T$; cf. \eqref{eq:inv_vperp}.

\smallskip
By assumption the subspace $V$ is invariant under both
$RD = X + \id$ and $(RD)^T = DR^T = X^T + \id$. Therefore
$$
D^2V = (DR^T)(RD) V \subseteq (DR^T)V \subseteq V\;.
$$
From the assumption on $D$, we have $DV \subseteq V$. Since $D$
has only nonzero eigenvalues $D$ is invertible and so $DV = V$.
It follows that
$$
RD V \subseteq V \quad\implies\quad RV \subseteq V\;.
$$
Since $R \in \SO(n)$ is invertible, we have $RV = V$. Reversing the roles
of $V$ and $V^\perp$, we can apply the same argument to $V^\perp$.
\end{proof}

By~\theref{theo:real_square_roots}, as phrased
in~\remref{rem:real_square_roots}, there exists a sequence
of pairwise orthogonal vector spaces $V_i$, $i = 1,\ldots,r$,
with $1 \leq \dim V_i \leq 2$ which decompose
$\R^n = V_1 \oplus_\perp V_2 \oplus_\perp \ldots \oplus_\perp V_r$.
These correspond to a block-diagonal representation of
$X(R) \eqdef RD - \id$. The existence of an associated orthogonal
change of basis matrix $T \in O(n)$ is also assured
by~\theref{theo:real_square_roots}.
Furthermore, by \leref{lem:sim_inv}, both $R$ and $D$ are also
block-diagonal with respect to this choice of basis.
This means, in particular, that \emph{any solution} $R$ satisfying
the symmetric square condition
$(X(R))^2 = (RD - \id)^2 \in \Sym(n)$
admits a block-diagonal representation. Since this condition
characterizes the critical points
by \leref{lem:symmetric_square_condition},
any critical point of $W(R\,;D)$ admits a representation
\begin{equation}
  \label{eq:critical_point_block_structure}
  \widetilde{R}
  = T^{-1}RT
  = \diag(\widetilde{R}_1,\ldots,\widetilde{R}_r)
  = \begin{pmatrix}
    \widetilde{R}_1  & 0               & \ldots & 0\\
    0                & \widetilde{R}_2 & \ddots & \vdots\\
    \vdots           & \ddots          & \ddots & 0\\
    0                & \ldots          & 0      & \widetilde{R}_r
  \end{pmatrix} \in O(n) \subset \R^{n \times n}
\end{equation}
in block-diagonal form. Here, the blocks on the diagonal satisfy
$\widetilde{R}_i \in O(n_i)$, $i = 1,\ldots,r$, with
$n_i \in \{1,2\}$ and $\sum_i^r n_i = n$.

In the basis provided by $T \in \O(n)$, any critical point
$R \in \O(n)$ can be constructed from solutions
$\widetilde{R}_i \in \O(n_i)$ of one- and two-dimensional
subproblems
\begin{equation}
  \left(\widetilde{X}(\widetilde{R}_i)\right)^2 \in \Sym(n_i)\;.
\end{equation}
Note that these subproblems are now posed on the space of
\emph{orthogonal}, rather than \emph{special orthogonal}
matrices.

\begin{ass}
  \label{ass:Dpos}
  For the purpose of clarity of exposition, we make an additional,
  stronger assumption on the diagonal matrix
  $D = \diag(d_1,\ldots,d_n)$, namely
  $$
  d_1 > d_2 > \ldots > d_n > 0\;.
  $$
\end{ass}
The slightly more general case of possibly non-distinct positive
entries $d_i$ can be treated similarly which we will indicate in
running commentary.

\begin{rem}[Implications of $D$-invariance]
  Under the \assref{ass:Dpos}, the $D$-invariance of the
  subspaces $V_i$ shown in \leref{lem:sim_inv} implies a strong
  restriction: the $V_i$ are necessarily coordinate subspaces in
  the standard basis of $\R^n$. Thus, we can index these data by
  partitions of the index set $\{1,\ldots,n\}$ into disjoint subsets
  of size one or two. Furthermore, by picking a standard coordinate
  basis for each $V_i$, we can ensure that the change of basis matrix
  $T \in \O(n)$ is a permutation matrix.
\end{rem}

We summarize that this particular structure allows to reduce
the optimization~\probref{prob:opt} to a finite list of
decoupled one- and two-dimensional subproblems.
However, we have to consider minimization
with respect to \emph{orthogonal} matrices $R \in \O(n)$
instead of $R \in \SO(n)$. This will be the content of
the next section.
 \countres
\makeatletter{}\section{Analysis of the decoupled subproblems}
\label{sec:subproblems}
Let $I \subseteq \{1,\ldots,n\}$ be a one-element subset $\{i\}$
or a two-element subset $\{i,j\}$ and let $D_I$ be the associated
restriction of $D$ given by
$$
\begin{cases}
D_I \eqdef \begin{pmatrix} d_i \end{pmatrix}, &\text{if}\quad I = \{i\},\\[.5em]
D_I \eqdef \begin{pmatrix} d_i & 0 \\ 0 & d_j\end{pmatrix}, &\text{if}\quad I = \{i,j\}\;.
\end{cases}
$$

In this section we solve for critical points of the function
$$
W(R_I\,;D_I) \eqdef \hsnorm{\sym(R_ID_I -\id)}^2
$$
for $R_I \in O(\abs{I})$ and compute the corresponding
critical values. This corresponds to the solution of the decoupled
lower-dimensional subproblems as described in the previous section.

\begin{theo}[Critical points: size one]
  \label{theo:size_one}
  For $I = \{i\}$ we have the submatrix $D_I = (d_i)$ and
  $R_I = \pm \id = (\pm 1)$. The realized critical energy
  levels are
  \begin{equation}
    W(+\id\,;D_I) = (d_i - 1)^2
    \quad\text{and}\quad
    W(-\id\,;D_I) = (d_i + 1)^2\;.
  \end{equation}
\end{theo}

\begin{proof}
  There are only two orthogonal matrices in dimension one and
  the result is immediate.
\end{proof}

For the case $\abs{I} = 2$, we consider the two separate cases
$\det{R_I} = 1$ and $\det{R_I} = -1$.

\begin{theo}[Critical points: size two and positive determinant]
  \label{theo:size_two_pos}
  The critical points $R_I$ with $\det{R_I} = 1$ are described
  as follows. For any values $d_i$ and $d_j$
  the matrices $R_I = \pm \id$ are critical points with the
  critical values $(d_i - 1)^2 + (d_j - 1)^2$
  and $(d_i + 1)^2 + (d_j + 1)^2$, respectively.
  In addition, if $d_i + d_j > 2$, there are two non-diagonal
  critical points
  \begin{equation}
    R_I = \begin{pmatrix}
      \cos \alpha  & -\sin \alpha \\
      \sin \alpha  &  \cos \alpha
    \end{pmatrix},\quad\text{with}\quad \cos \alpha = \frac{2}{d_i + d_j}
  \end{equation}
  which attain the same critical value
  \begin{equation}
    W(R_I\,;D_I) = \frac{1}{2} (d_i - d_j)^2\;.
  \end{equation}
\end{theo}
\begin{proof}
  By \leref{lem:symmetric_square_condition} $R_I$ is a critical
  point if and only if $(R_ID_I - \id)^2$ is symmetric.
  We may thus apply \leref{lem:root_2x2} which implies
  $R_ID_I - \id \in \Sym(2)$ or $\tr{R_ID_I - \id} = 0$.
  Using the explicit representation
  $$
  R_I = \begin{pmatrix}
    \cos \alpha  & -\sin \alpha \\
    \sin \alpha  &  \cos \alpha
  \end{pmatrix}\;,
  $$
  the symmetry condition $R_ID_I - \id \in \Sym(2)$
  is equivalent to $(d_i + d_j)\sin\alpha = 0$ which has
  two solutions $R_I = \pm \id$. The trace condition
  $\tr{R_ID_I - \id} = 0$ is equivalent to
  $(d_i + d_j)\cos\alpha = 2$ which can be solved for
  $\alpha$ if and only if $d_i + d_j \geq 2$.
  It gives rise to two non-diagonal solutions if and
  only if $d_i + d_j > 2$.\\
  \smallskip
  In the first case $R_I = \pm \id$, the critical values
  are immediately seen to be $(d_i - 1)^2 + (d_j - 1)^2$
  and $(d_i + 1)^2 + (d_j + 1)^2$, respectively.\\
  \smallskip
  In the second case, the critical values are calculated
  as follows. Observing that
  \begin{equation}
    \sym(R_I D_I - \id) =
    \begin{pmatrix}
      d_i \cos \alpha - 1 & \frac{1}{2}(d_j - d_i)\sin \alpha\\
      \frac{1}{2}(d_j - d_i)\sin \alpha &  d_j \cos \alpha - 1
    \end{pmatrix}
  \end{equation}
  we use $(d_i + d_j)\cos\alpha = 2$ to get
  \begin{align}
    \hsnorm{\sym(R_I D_I - \id)}^2
    &= (d_i \cos \alpha - 1)^2
    + (d_j \cos \alpha - 1)^2
    + \frac{1}{2}(d_j - d_i)^2 \sin^2 \alpha \notag\\
    &= (d_i^2 + d_j^2)\cos^2 \alpha - 2(d_i +d_j)\cos \alpha
    + 2 + \frac{1}{2} (d_j - d_i)^2 (1 - \cos^2 \alpha) \notag\\
    &=\frac{1}{2} (d_j - d_i)^2
    + \frac{1}{2}(d_i + d_j)^2 \cos^2\alpha
    - 2(d_i +d_j)\cos \alpha + 2 \notag\\
    &= \frac{1}{2} (d_i - d_j)^2 + 2 - 4 + 2
    =  \frac{1}{2} (d_i - d_j)^2\;.
  \end{align}
  This shows the claim.
\end{proof}

\begin{theo}[Critical points: size two and negative determinant]
  \label{theo:size_two_neg}
  The critical points $R_I$ with $\det{R_I} = -1$ are described
  as follows. For any values $d_i$ and $d_j$
  the diagonal matrices $R_I = \pm \diag(1,-1)$ are critical
  points with the critical values $(d_i - 1)^2 + (d_j + 1)^2$
  and $(d_i + 1)^2 + (d_j - 1)^2$, respectively.
  In addition, for $\abs{d_i - d_j} > 2$, there are two
  non-diagonal critical points
  \begin{equation}
    R_I = \begin{pmatrix}
      \cos \alpha  & \sin \alpha \\
      \sin \alpha  & -\cos \alpha
    \end{pmatrix},\quad\text{with}\quad \cos \alpha = \frac{2}{\abs{d_i - d_j}}\;,
  \end{equation}
  which attain the same critical value
  \begin{equation}
    W(R_I\,;D_I) = \frac{1}{2} (d_i + d_j)^2\;.
  \end{equation}
\end{theo}

\begin{proof}
  By \leref{lem:symmetric_square_condition} $R_I$ is a critical
  point if and only if $(R_ID_I - \id)^2$ is symmetric.
  We may thus apply \leref{lem:root_2x2} which implies
  $R_ID_I - \id \in \Sym(2)$ or $\tr{R_ID_I - \id} = 0$.
  Using the explicit representation
  $$
  R_I = \begin{pmatrix}
    \cos \alpha  &  \sin \alpha \\
    \sin \alpha  & -\cos \alpha
  \end{pmatrix}
  $$
  the symmetry condition $R_ID_I - \id \in \Sym(2)$
  is equivalent to
  \begin{equation}
    (d_i - d_j)\sin\alpha = 0
  \end{equation}
  which has two solutions $R_I = \pm \diag(1,-1)$ since
  $d_i \neq d_j$ due to \assref{ass:Dpos}. The trace condition
  $\tr{R_ID_I - \id} = 0$ is equivalent to
  $(d_i - d_j)\cos\alpha = 2$ which can be solved for
  $\alpha$ if and only if $\abs{d_i - d_j} \geq 2$.
  Thus there are two non-diagonal solutions if and
  only if $\abs{d_i - d_j} > 2$.\\
  \smallskip
  In the first case $R_I = \pm \diag(1,-1)$, the critical values
  are immediately seen to be $(d_i - 1)^2 + (d_j + 1)^2$
  and $(d_i + 1)^2 + (d_j - 1)^2$, respectively.\\
  \smallskip
  In the second case, the critical values are calculated
  as follows. Observing that
  \begin{equation}
    \sym(R_I D_I - \id) =
    \begin{pmatrix}
      d_i \cos \alpha - 1               & \frac{1}{2}(d_i + d_j)\sin \alpha\\
      \frac{1}{2}(d_i + d_j)\sin \alpha &  -d_j \cos \alpha - 1
    \end{pmatrix}
  \end{equation}
  we use $\abs{d_i - d_j}\cos\alpha = 2$ to get
  \begin{align}
    \hsnorm{\sym(R_I D_I - \id)}^2
    &= (d_i \cos \alpha - 1)^2
    + (d_j \cos \alpha + 1)^2
    + \frac{1}{2}(d_i + d_j)^2 \sin^2 \alpha \notag\\
    &= (d_i^2 + d_j^2)\cos^2 \alpha - 2(d_i - d_j)\cos \alpha
    + 2 + \frac{1}{2} (d_i + d_j)^2 (1 - \cos^2 \alpha) \notag\\
    &=\frac{1}{2} (d_i + d_j)^2
    + \frac{1}{2}(d_i - d_j)^2 \cos^2\alpha
    - 2(d_i - d_j)\cos \alpha + 2 \notag\\
    &= \frac{1}{2} (d_i + d_j)^2 + 2 - 4 + 2
    =  \frac{1}{2} (d_i + d_j)^2\;.
  \end{align}
  This shows the claim.
\end{proof}

\begin{rem}[The positive choice $\det{R_I} = +1$ minimizes energy]
  \label{rem:pos_det_minimal}
  A direct comparison of the energy levels realized by
  the different choices for the determinant of $R_I$ is
  instructive. Summarizing our preceding results, we have
  for $\abs{I} = 1$, i.e., for a block of size one
  \begin{align}
    \det{R_I} \,=\, +1 &\quad\quad\mapsto\quad\quad (d_i - 1)^2\;,\\
    \det{R_I} \,=\, -1 &\quad\quad\mapsto\quad\quad (d_i + 1)^2
    \quad\geq\quad (d_i - 1)^2\;.
  \end{align}
  Similarly, for $\abs{I} = 2$, i.e., for a block of size two, we obtain
  \begin{align}
    \det{R_I} \,=\, +1 &\quad\quad\mapsto\quad\quad \frac{1}{2}\,(d_i - d_j)^2\;,\\
    \det{R_I} \,=\, -1 &\quad\quad\mapsto\quad\quad \frac{1}{2}\,(d_i + d_j)^2
    \quad\geq\quad \frac{1}{2}\,(d_i - d_j)^2\;.
  \end{align}
  The estimates follow from our \assref{ass:Dpos} on the entries $d_i > 0$
  of the diagonal matrix $D > 0$.
\end{rem}

\begin{rem}
  \label{rem:reducible_blocks}
  The diagonal critical points $R_I = \pm \id$
  and $R_I = \pm \diag(1,-1)$ reduce to size one blocks
  (or index subsets $\abs{I} = 1$) in the block
  decomposition \eqref{eq:critical_point_block_structure}.
\end{rem}

\begin{rem}[On non-distinct entries of $D$]
  If we relax the \assref{ass:Dpos} and allow for
  $$
  d_1 \geq d_2 \geq \ldots \geq d_n > 0
  $$
  then there are degenerate critical points with
  $\det{R_I} = -1$ if and only if $d_i = d_j$.
  The corresponding critical value is the same as
  that realized by the diagonal matrices
  $\pm \diag(1,-1)$.
\end{rem}
 \countres
\makeatletter{}\section{Global minimization of the Cosserat shear-stretch energy}
\label{sec:optimality}
Combining the results of the two preceding sections, we can now
describe the critical values of the Cosserat shear-stretch energy
$W(R\,;D)$ which are attained at the critical points.  The main result
of this section is a procedure (algorithm) which traverses the set of
critical points in a way that reduces the energy at every step of the
procedure and finally terminates in the subset of global minimizers.

Technically, we label the critical points by certain partitions
of the index set $\{1,\ldots,n\}$ containing only subsets $I$ with one
or two elements. In the last section, we have seen that the subsets
$I$ and a choice of sign for $\det{R_I}$ uniquely characterize a critical
point $R \in \SO(n)$.

The next theorem expresses the value of $W(R\,;D)$ realized
by a critical point in terms of the labeling partition and choice of
determinants $\det{R_I}$ which characterize it.

\begin{theo}[Characterization of critical points and values]
  \label{theo:critical_values}
  Let $D \eqdef \diag(d_1,\ldots,d_n) > 0$ satisfy \assref{ass:Dpos}, i.e.,
  $d_1>d_2>\ldots>d_n>0$. Then the critical points $R \in \SO(n)$ of
  the objective function
  $$
  W(R\,;D) \eqdef \hsnorm{\sym(RD - \id)}^2
  $$
  can be classified according to partitions of the index set
  $\{1,\ldots,n\}$ into subsets of size one or two and
  choices of signs for the determinant $\det{R_I}$
  for each subset $I$. The subsets of size two
  $I = \{i,j\}$ satisfy
  $$
  \begin{cases}
  \phantom{|} d_i + d_j \phantom{|} > 2,     & \quad \det{R_I} = +1\;,\quad\text{and}\\
  \abs{d_i - d_j} > 2, & \quad \det{R_I} = -1\;.
  \end{cases}
  $$
  The critical values are given by
  $$
  W(R\,;D) =   \sum_{\substack{I = \{i\} \\ \det{R_I} = 1}} (d_i - 1)^2
  + \sum_{\substack{I = \{i\}   \\ \det{R_I} = -1}} (d_i + 1)^2
  + \sum_{\substack{I = \{i,j\} \\ \det{R_I} = 1}}  \frac{1}{2}(d_i - d_j)^2
  + \sum_{\substack{I = \{i,j\} \\ \det{R_I} = -1}} \frac{1}{2}(d_i + d_j)^2\;.
  $$
\end{theo}
\begin{proof}
  A suitable partition of the index set $\{1,\ldots,n\}$
  can be constructed as detailed in \secref{sec:critical_points}.
  The contributions of the subsets $I$ of size one and two are
  given by the theorems of \secref{sec:subproblems}. It suffices
  to consider the non-diagonal critical points for the subproblems
  of size two, because the diagonal cases can be accounted for by
  splitting the subset $I = \{i,j\}$ into two subsets $\{i\}$ and
  $\{j\}$  of size one, see~\remref{rem:reducible_blocks}.
\end{proof}

\begin{rem}[On non-distinct entries of $D$]
  If we relax the \assref{ass:Dpos} and allow for
  $$
  d_1 \geq d_2 \geq \ldots \geq d_n > 0
  $$
  then the $D$- and $R$-invariant subspaces $V_i$ are not
  necessarily coordinate subspaces. This produces non-isolated
  critical points but does not change the formula for the
  critical values.
\end{rem}

It seems instructive to precede our further development with an
outline of the scheme which allows us to traverse the set of critical
points such that the energy decreases in every step and terminates in
a global minimizer. Note that the scheme is conveniently formulated
in terms of the labeling partitions which classify the critical points:
\begin{scheme}[Construction of a minimizing sequence of critical points]
  Starting from the labeling partition of an arbitrary critical point:
  \begin{enumerate}
  \item Choose the positive sign $\det{R_I} = +1$ for each subset
    of the partition (cf.~\remref{rem:pos_det_minimal}
    and~\remref{rem:pos_det_admiss}).
  \item Disentangle all overlapping blocks for $n > 3$
    (cf.~\leref{lem:overlap}).
  \item Successively shift all $2\times2$-blocks to the lowest
    possible index, i.e., collect the blocks of size two
    as close to the upper left corner of the matrix $R$
    as possible (cf.~\leref{lem:comparison}).
  \item Introduce as many additional $2\times 2$-blocks by joining
    adjacent blocks of size $1$ as the constraint $d_i + d_j > 2$
    allows (cf.~\leref{lem:comparison}).
  \end{enumerate}
\end{scheme}
At the end of this section, we provide an \exaref{exa:algo}.

In order to compute the global minimizers $R \in \SO(n)$ for the
Cosserat shear-stretch energy $W(R\,;D)$, we have to compare
all the critical values which correspond to the different
partitions and choices of the signs of the determinants
in the statement of \theref{theo:critical_values}. In what follows,
we prove the reduction steps of the preceding scheme.

\begin{rem}
\label{rem:pos_det_admiss}
Notice that under \assref{ass:Dpos}, we have that
$\abs{d_i - d_j} > 2$ implies that $d_i + d_j > 2$. Therefore,
it is always possible to replace negative determinant choices
by positive ones. In the process the value of $W(R\,;D)$ is
reduced. Therefore, if $R$ is a
critical point which is a global minimizer of $||\sym(RD-\id)||^2$,
it only contains $R_I$ with determinant $\det{R_I} = 1$.
\end{rem}
This allows us to assume that $\det{R_I} = 1$ for all subsets
$I$ without any loss of generality.

The following lemma shows that blocks of size two are always
favored \emph{whenever they exist}.

\begin{lem}[Comparison lemma]
  \label{lem:comparison}
  If $d_i + d_j > 2$ then the difference between
  the critical values of $W(R\,;D)$ corresponding to the
  choice of a size two subset $I = \{i,j\}$ as compared
  to the choice of two size one subsets $\{i\}$, $\{j\}$
  is given by
  $$
  -\frac{1}{2}(d_i+d_j-2)^2.
  $$
\end{lem}
\begin{proof}
We subtract the corresponding contributions of the subsets and
simplify
$$
\frac{1}{2}(d_i - d_j)^2 - (d_i-1)^2 - (d_j-1)^2
  = -\frac{1}{2}(d_i + d_j - 2)^2\;.
  $$
This proves the claim.
\end{proof}
\medskip

Let us rewrite $W(R\,;D)$ in a slightly different form in order to
distill the contributions of the size two blocks in the partition.

\begin{cor}
  \label{cor:savings}
For the choices of $\det{R_I} = 1$ there holds
$$
W(R\,;D) = \hsnorm{\sym(RD-\id)}^2 = \sum_{i=1}^n(d_i - 1)^2
                  -\frac{1}{2}\sum_{I = \{i,j\}} (d_i+d_j-2)^2.
$$
\end{cor}

\begin{proof}
  The first term in the formula is the value realized by $W(R\,;D)$
  for the trivial partition into $n$ subsets of size one. By virtue
  of the Comparison Lemma~\ref{lem:comparison} each block of size
  two reduces the critical value by the amount
  $\frac{1}{2}(d_i + d_j-2)^2$.
\end{proof}

Let us now consider the case of dimension $n = 3$ explicitly in order
to prepare the exposition of the higher dimensional case.
\begin{theo}
  \label{theo:global_min_3d}
  Let $d_1 > d_2 > d_3 > 0$. If $d_1 + d_2 \leq 2$ then the
  global minimum of
  $$
  W(R\,;D) \eqdef \hsnorm{\sym(RD-\id)}^2
  $$
  occurs at $R = \id$ and is given by
  $$
  W(R\,;D) = (d_1 - 1)^2 + (d_2 - 1)^2 + (d_3 - 1)^2 \;.
  $$
  If $d_1 + d_2 > 2$ then the global minimum is realized by
  either of two critical points of the form
  $$
  R = \begin{pmatrix}
    \cos \alpha & -\sin\alpha & 0\\
    \sin \alpha &  \cos\alpha & 0\\
    0           & 0           & 1
  \end{pmatrix}
  \quad\text{with}\quad (d_1 + d_2)\cos \alpha = 2\;.
  $$
  In this case the global minimum is
  $$
  W(R\,;D) = (d_1 - 1)^2 + (d_2 - 1)^2 + (d_3 - 1)^2 - \frac{1}{2}(d_1 + d_2 - 2)^2
  = \frac {1}{2}(d_1 - d_2)^2 + (d_3 - 1)^2\;.
  $$
\end{theo}
\begin{proof}
  If $d_1 + d_2 \leq 2$ then $d_i + d_j \leq 2$ for all index pairs
  $(i,j)$ and there are no blocks of size two at the global minimum.
  If $d_1 + d_2 > 2$ then the choice of partition $\{1,2\} \sqcup \{3\}$
  is admissible. \coref{cor:savings} shows that this is always favorable
  compared to the partition into three size one subsets
  $\{1\} \sqcup \{2\} \sqcup \{3\}$.
  Whether or not other size two subsets are admissible according to
  the inequalities $d_i + d_j > 2$, the partition
  $\{1,2\} \sqcup \{3\}$ is always optimal. This follows from the
  ordering $d_1 > d_2 > d_3 > 0$ which implies that the
  partition-dependent term
  $\frac{1}{2}(d_i + d_j -2)^2$ in \coref{cor:savings} is maximized
  for $I = \{i,j\} = \{1,2\}$.
\end{proof}

In general, a deformation gradient $F \in \GL^+(n)$ can have non-distinct
singular values $\nu_i = \nu_j$, $i \neq j$. This situation may arise,
e.g., due to a symmetry assumption in mechanics.

\begin{rem}[On non-distinct entries of $D$]
  Assume $d_1 \geq d_2 \geq d_3 > 0$. Our results imply the following:\\
  \medskip
  If $d_1 + d_2 \leq 2$, then all $V_i$ are of dimension $1$. Since
  the restriction of a given minimizer $R$ to each $V_i$ satisfies
  $R|_{V_i} = \id$, we see that $R = \id$. The global minimum of
  the Cosserat shear-stretch energy is given by
  $$
  W(R\,;D) = (d_1 - 1)^2 + (d_2 - 1)^2 + (d_3 - 1)^2\;.
  $$
  If $d_1 + d_2 > 2$, then for a global minimizer $R$ there is
  a one-dimensional $R$-invariant subspace which is also $D$-invariant
  with associated eigenvalue $d_3$. Therefore, $R$ is a rotation with
  axis in the $d_3$-eigenspace of $D$. The rotation angle satisfies
  the relation $(d_1 + d_2)\cos \alpha = 2$ and the global minimum
  of the energy is given by
  $$
  W(R\,;D) = (d_1 - 1)^2 + (d_2 - 1)^2 + (d_3 - 1)^2 - \frac{1}{2}(d_1 + d_2 - 2)^2
  = \frac {1}{2}(d_1 - d_2)^2 + (d_3 - 1)^2\;.
  $$
  This case further splits into several subcases all realizing the
  same energy level according to the multiplicity of
  the eigenvalue $d_3$:\\
  \smallskip
  If $d_1 \geq d_2 > d_3$, i.e., the multiplicity of $d_3$ is one,
  then there are two isolated global minimizers which are rotations
  with rotation angle $\arccos(2/ (d_1 + d_2))$ with respect
  to either of the two half-axes in $\vspan{e_3}$ (as in the
  case of distinct entries of $D$ discussed
  in \theref{theo:global_min_3d}).\\
  \smallskip
  If $d_1 > d_2 = d_3$, i.e., the multiplicity of $d_3$ is two,
  then the global minimizers $R$ form a one-dimensional family
  of rotations with rotation angle $\arccos(2/(d_1 + d_2))$ and
  rotation half-axes in the $d_3$-eigenplane $\vspan{e_2,e_3}$
  of $D$.\\
  \smallskip
  If $d_1 = d_2 = d_3$, i.e., the multiplicity of $d_3$ is three,
  then there is a two-dimensional family of global minimizers
  $R$ which are rotations with rotation angle
  $\arccos(2/(d_1 + d_2))$ about arbitrary half-axes in $\R^3$.
\end{rem}
It is interesting that the set of global minimizers is
\emph{connected} in the last two cases where $d_2 = d_3$. This
allows for a continuous transition between minimizers with opposite
half-axes which are inverses of each other.

To study the global minimizers for the Cosserat shear-stretch
energy in arbitrary dimension $n \geq 4$, we need to investigate
the relative location of the size two subsets of the partition.

\begin{lem}
  \label{lem:overlap}
  Let $R \in \SO(n)$ be a global minimizer for $W(R\,;D)$.
  Then $R$ cannot contain overlapping size two subsets, i.e.,
  $I = \{i_1,i_4\}$, $J = \{i_2,i_3\}$, with $i_1 < i_2 < i_3 < i_4$.
\end{lem}

\begin{proof}
  We assume that $R$ is a global minimizer corresponding to a partition
  containing two overlapping subsets as described above and derive a
  contradiction.\\
  It suffices to consider the case $i_1 = 1, i_2 = 2, i_3 = 3$ and
  $i_4 = 4$ with the general case being completely analogous. We recall
  the ordering $d_1 > d_2 > d_3 > d_4 > 0$.
  \\
  \smallskip
  There are two cases to consider:\\[2em]
  \medskip
  {\bf Case 1:} $d_3 + d_4 > 2$. In this case, we can consider
  another critical point $\mathring{R}$ corresponding to the
  partition $\{1,2\} \sqcup \{3,4\}$ instead
  of $\{1,4\} \sqcup \{2,3\}$. By \coref{cor:savings} we have
  \begin{align*}
    W(R\,;D) - W(\mathring{R}\,;D)
    &=
    \frac{1}{2} (d_1 + d_2-2)^2
  + \frac{1}{2} (d_3 + d_4-2)^2
  - \frac{1}{2} (d_1 + d_4-2)^2
  - \frac{1}{2} (d_2 + d_3-2)^2\\
  &= d_1d_2+d_3d_4 - d_1d_4-d_2d_3
   = (d_1-d_3)(d_2-d_4) > 0.
  \end{align*}
  Thus $R$ is not a global minimum of $W(R\,;D)$.\\[2em]
  {\bf Case 2:} $d_3 + d_4 \leq 2$. In this case, we can not have
    the size two subset $\{3,4\}$. However, it is possible
    to decrease the value of $W(R\,;D)$ by choosing another
    critical point $\mathring{R}$ corresponding to the partition
    $\{1,2\} \sqcup \{3\} \sqcup \{4\}$ instead of
    $\{1,4\} \sqcup \{2,3\}$. By \coref{cor:savings} we have
  \begin{align*}
    W(R\,;D) - W(\mathring{R}\,;D)
    &=
    \frac{1}{2} (d_1 + d_2-2)^2
  - \frac{1}{2} (d_1 + d_4-2)^2
  - \frac{1}{2} (d_2 + d_3-2)^2\\
  &\geq \frac{1}{2} (d_1 + d_2-2)^2
  - \frac{1}{2} (d_1 + (2 - d_3)-2)^2
  - \frac{1}{2} (d_2 + d_3-2)^2\\
  &= \frac{1}{2}(d_1 + d_2 - 2)^2 - \frac{1}{2}(d_1 - d_3)^2 - \frac{1}{2}(d_2 + d_3-2)^2\\
  & =(d_1 - d_3)(d_2 + d_3 - 2) > 0.
  \end{align*}
  In the first inequality we use the fact that for $d_1 + d_4 \geq 2$
  the function $(d_1 + d_4 - 2)^2$ is increasing in $d_4$ and
  $d_4 \leq 2 - d_3$ by assumption. This shows that $R$ is not a global
  minimum of $W(R\,;D)$.\\
  \medskip
  We arrive at a contradiction in both cases which proves the statement.
\end{proof}

\medskip
We are now ready to state and prove the general $n$-dimensional case.

\begin{theo}
  \label{theo:global_min_nd}
  Let $D \eqdef \diag(d_1,\ldots,d_n) > 0$ with ordered entries
  $d_1 > d_2 > \ldots > d_n > 0$. Let us fix the maximum $k \in \N_0$
  for which $d_{2k-1} + d_{2k} > 2$. Any global minimizer $R \in \SO(n)$ of
  $$
  W(R\,;D) \eqdef \hsnorm{\sym(RD - \id)}^2
  $$
  corresponds to a partition of the index set $\{1,\ldots,n\}$ with
  $k \geq 0$ leading subsets of size two
  $$
  \underbrace{\{1,2\} \sqcup \{3,4\} \sqcup \ldots \sqcup \{2k-1, 2k\}}_{k\;\rm{subsets\,of\,size\,two}}
  \;\sqcup\;
  \underbrace{\{2k+1\} \sqcup \ldots \sqcup \{n\}}_{(n - 2k)\;\text{subsets\,of\,size one}}
  $$
  in the classification of critical points provided by \theref{theo:critical_values}. The global minimum of $W(R\,;D)$ is given by
  \begin{align*}
    W^{\rm red}(D) \eqdef& \min_{R \in \SO(n)}{W(R\,;D)}
    = \sum_{i=1}^n(d_i-1)^2 - \frac{1}{2}\sum_{i=1}^k (d_{2i-1}+d_{2i}-2)^2\\
    =& \frac{1}{2}\sum_{i=1}^k (d_{2i-1} - d_{2i})^2 + \sum_{i=2k+1}^n (d_i-1)^2\;.
  \end{align*}
\end{theo}

\begin{proof}
  \leref{lem:overlap} shows that a global minimizer $R \in \SO(n)$
  can not have a partition with \emph{overlapping} size two subsets.
  As in the proof of \theref{theo:global_min_3d} (the $n = 3$ case)
  we can decrease the value of $W(R\,;D)$ by shifting down
  the indices of all size two subsets as far as possible.
  Therefore the optimal partition is of the form
  $$
  \{1,2\} \sqcup \{3,4\} \sqcup \ldots \sqcup \{2l-1, 2l\}
  \sqcup \{2l+1\} \sqcup \ldots \sqcup \{n\}
  $$
  for some $l \leq k$. By \coref{cor:savings} the global
  minimum is realized by the critical points corresponding
  to the maximal possible choice $l = k$. The value of
  $W(R\,;D)$ at a global minimizer is computed by inserting
  the corresponding optimal partition into
  \theref{theo:critical_values} and \coref{cor:savings}.
\end{proof}

\begin{rem}
  The number of global minimizers in the above theorem is
  $2^k$, where $k$ is the number of blocks of size two in
  the preceding characterization of a global minimizer as
  a block-diagonal matrix. All global minimizers are
  block-diagonal similar to the $n = 3$ case (\theref{theo:global_min_3d}).
\end{rem}

\begin{rem}[On non-distinct entries of $D$]
  If we relax the \assref{ass:Dpos} and allow for
  $$
  d_1 \geq d_2 \geq \ldots \geq d_n > 0
  $$
  then the global minimizers may or may not be isolated.
  The formula for the reduced energy as stated in
  \theref{theo:global_min_nd}
  is, however, not affected.
\end{rem}

The following example illustrates the energy-minimizing traversal
of critical points which always terminates in a global minimizer.
\begin{exa}
  \label{exa:algo}
Let $D = \diag\left(4,2,1,\frac{1}{2},\frac{1}{4}\right)$.
\theref{theo:critical_values} shows that the critical points
can be characterized by certain partitions\footnote{More precisely,
  a labeling partition uniquely characterizes sets of critical points
which generate the same critical value. A block of size two, for
example, characterizes two different symmetric solutions corresponding
to the choice of sign for the rotation angle $\alpha$. Both choices,
however, yield the same value for the energy.}
of the index set $\{1,2,3,4,5\}$ and a choice of a sign for each
subset $I$ of the partition. Thus, we introduce the convenient notation
of a pair of a subset and a sign $(I,\pm)$, where the sign encodes a
possible choice for the determinant $\det{R_I} = \pm 1$.\\

{\bf Setup:} We consider a critical point $R^{(0)}$ corresponding to
the labeling partition
  \begin{equation}
    \mathcal{P}^{(0)} =
    \left\{
    \left(\{1\}, \,+\,\right),\,
    \left(\left\{2, 5\right\}, \,-\,\right),\,
    \left(\{3\}, \,-\,\right),\,
    \left(\{4\}, \,-\,\right)
    \right\}\;.
  \end{equation}
  Note that $d_2 + d_5 = 2 + \frac{1}{4} > 2$, i.e., the
  $2 \times 2$-block corresponding to $I = \{2,5\}$ exists,
  as required for a valid partition characterizing a critical
  point $R^{(0)}$. The corresponding critical value of the summation
  formula in the statement of \theref{theo:critical_values} is
  given by
  \begin{align}
    W^{(0)} &= W(R^{(0)}\,;D)
    = \underbrace{(4 - 1)^2}_{\left(\{1\}, \,+\,\right)}
    + \underbrace{(1 + 1)^2}_{\left(\{3\}, \,-\,\right)}
    + \underbrace{\left(1 + \frac{1}{2}\right)^2}_{\left(\{4\}, \,-\,\right)}
    + \underbrace{\frac{1}{2}\left(2 + \frac{1}{4}\right)^2}_{\left(\left\{2, 5\right\}, \,-\,\right)}\\
    &= \frac{569}{32} \approx 17.78\;.\notag
  \end{align}

  {\bf Step 1 (Choice of positive sign):}
  We consistently choose the positive sign for the determinant
  in the labeling partition which gives
  \begin{equation}
    \mathcal{P}^{(1)} =
    \left\{
    \left(\left\{1, 5\right\}, \,+\,\right),\,
    \left(\{2\}, \,+\,\right),\,
    \left(\{3\}, \,+\,\right),\,
    \left(\{4\}, \,+\,\right)
    \right\}\;.
  \end{equation}
  This updated partition characterizes a \emph{different}
  critical point $R^{(1)}$ realizing a lower energy level
  \begin{align}
    W^{(1)} &= W(R^{(1)}\,;D) =
     \underbrace{(4 - 1)^2}_{\left(\{1\}, \,+\,\right)}
    + \underbrace{(1 - 1)^2}_{\left(\{3\}, \,+\,\right)}
    + \underbrace{\left(1 - \frac{1}{2}\right)^2}_{\left(\{4\}, \,+\,\right)}
    + \underbrace{\frac{1}{2}\left(2 - \frac{1}{4}\right)^2}_{\left(\left\{2, 5\right\}, \,+\,\right)}\\
    &= \frac{345}{32} \approx 10.28\;.\notag
  \end{align}
  {\bf Step 2 (Disentanglement):}
  The next step of the procedure is to remove overlap of $2\times2$-blocks.
  In our example, we only have one such block and there is nothing to
  do, i.e., $\mathcal{P}^{(2)} = \mathcal{P}^{(1)}$.\\

  {\bf Step 3 (Index shift):}
  We now decrement the indices of the $2 \times 2$-blocks
  as much as possible, i.e., we string them together starting
  in the upper left corner. Shifting the $\{2,5\}$-block
  to $\{1,2\}$, we obtain the following new partition
  \begin{equation}
    \mathcal{P}^{(3)} =
    \left\{
    \left(\left\{1, 2\right\}, \,+\,\right),\,
    \left(\{3\}, \,+\,\right),\,
    \left(\{4\}, \,+\,\right),\,
    \left(\{5\}, \,+\,\right)
    \right\}\;.
  \end{equation}
  The energy level realized by a corresponding critical
  point $R^{(3)}$ is
  \begin{align}
    W^{(3)} &= W(R^{(3)}\,;D) =
    \underbrace{(1 - 1)^2}_{\left(\{3\}, \,+\,\right)}
    + \underbrace{\left(1 - \frac{1}{2}\right)^2}_{\left(\{4\}, \,+\,\right)}
    + \underbrace{\left(1 - \frac{1}{4}\right)^2}_{\left(\{5\}, \,+\,\right)}
    + \underbrace{\frac{1}{2}\left(4 - 2\right)^2}_{\left(\left\{1, 2\right\}, \,+\,\right)}\\
    &= \frac{45}{16} \approx 2.81\;.\notag
  \end{align}

{\bf Step 4 (Exhaustion by $2\times2$-blocks):}
In this step, we try to create as many $2\times 2$-blocks as possible.
We first locate the pair of subsets of size one with minimal indices
which is $(\{3\},\{4\})$. Since $d_3 + d_4 = 1 + \frac{1}{2} \leq 2$,
no further $2\times2$-block exists. Thus,
$\mathcal{P}^{(4)} = \mathcal{P}^{(3)}$.\\

{\bf Result:} The finally obtained labeling partition
\begin{equation}
  \mathcal{P} = \mathcal{P}^{(4)} =
  \left\{
  \left(\left\{1, 2\right\}, \,+\,\right),\,
  \left(\{3\}, \,+\,\right),\,
  \left(\{4\}, \,+\,\right),\,
  \left(\{5\}, \,+\,\right)
  \right\}
\end{equation}
characterizes a global minimizer. With the notation of
\theref{theo:global_min_nd} the maximal number of $2\times2$-blocks
is $k = 1$ and we have $2^k = 2$ global minimizers of the form
  \begin{equation}
    \rpolar(D) = \left(\;
    \begin{array}{ccccc}
      \cline{1-2}
      \multicolumn{1}{|c}{\cos\alpha_1} & -\sin\alpha_1 & \multicolumn{1}{|c}{0} & 0&0\\
      \multicolumn{1}{|c}{\sin\alpha_1} & \phantom{-}\cos\alpha_1 & \multicolumn{1}{|c}{0} &0 &0\\\cline{1-3}
      0&0 &\multicolumn{1}{|c}{1} &\multicolumn{1}{|c}{0} & 0\\
      \cline{3-4}
      0& 0 &0 &\multicolumn{1}{|c}{1} & \multicolumn{1}{|c}{0}\\
      \cline{4-5}
      0& 0 & 0&0 &\multicolumn{1}{|c|}{1}\\
      \cline{5-5}
    \end{array}
    \;\right)\;,\quad\text{with}\quad
    \cos(\alpha_{1}) = \frac{2}{d_1 + d_2} = \frac{1}{3}\;.
  \end{equation}
  Inserting the global minimizers into the energy, we obtain the reduced
  energy
  \begin{equation}
    W^{\rm red}(D) \eqdef W(\rpolar(D)\,;D) = \frac{45}{16} \approx 2.81\;.
  \end{equation}
  Just to give a comparison, the identity matrix $\id \in \SO(n)$
  realizes the energy level
  \begin{align}
    W(\id\,;D) &=
  \underbrace{(4 - 1)^2}_{\left(\{1\}, \,+\,\right)}
  + \underbrace{(2 - 1)^2}_{\left(\{2\}, \,+\,\right)}
  + \underbrace{(1 - 1)^2}_{\left(\{3\}, \,+\,\right)}
  + \underbrace{\left(1 - \frac{1}{2}\right)^2}_{\left(\{4\}, \,+\,\right)}
  + \underbrace{\left(1 - \frac{1}{4}\right)^2}_{\left(\{5\}, \,+\,\right)}\\
  &= \frac{173}{16} \approx 10.81\;.\notag
\end{align}
  Thus, the identity $\id \in \SO(n)$ is not a global minimizer.
\end{exa}

\begin{rem}[Optimality of $\id$]
  Our results imply that the identity matrix $\id \in \SO(n)$ is globally
  optimal for $W(R\,;D)$ with $D > 0$, if and only if there exists no
  $2\times2$-block with a positive choice of $\det{R_I}$, i.e.,
  $$
  \max_{1\leq i \neq j \leq n} (d_i + d_j) \quad\leq\quad 2\;.
  $$
  This corresponds to the tension-compression asymmetry described
  in~\cite{Fischle:2015:OC2D,Fischle:2015:OC3D,Fischle:2016:RPNI}
  for dimensions $n = 2,3$.
\end{rem}

 \countres
\makeatletter{}\section{Concluding remarks}
\label{sec:discussion}
For the sake of clarity of exposition, we have restricted our attention
to the case of a diagonal and positive definite parameter matrix
$D > 0$, i.e., $d_i > 0$. Our technical approach, however, readily
carries over to the more general case $d_i \neq 0$ with minor
modifications. The construction
\begin{equation}
  \hsnorm{\sym\left\{\left[R
    \left(
    \begin{array}{c|c}
      \id & \\\hline
      & -\id
    \end{array}
    \right)
    \right]
    \left[
    \left(
    \begin{array}{c|c}
      \id & \\\hline
      & -\id
    \end{array}
    \right)D
    \right]- \id\right\}
  }^2
\end{equation}
allows to reduce such a parameter matrix $D$ to
$\abs{D} \eqdef \diag(\abs{d_1},\ldots,\abs{d_n}) > 0$
which is positive definite. Note that the minimization must then
be carried out in the appropriate connected component of the
orthogonal matrices $\O(n)$. We also expect that the degenerate
case where some $d_i = 0$ can be handled with our techniques as
well.

The matrix group of rotations $\SO(3)$ equipped with its
natural bi-invariant Riemannian metric
\begin{equation}
  g(\xi, \eta)|_R \eqdef g(R^T\xi, R^T\eta)|_\id
                  \eqdef \scalprod{R^T \xi}{R^T \eta}
                  = \scalprod{\xi}{\eta}
\end{equation}
is a Riemannian manifold $(\SO(3), g)$. In~\cite{Neff01c}, the
dynamics of the following Riemannian gradient flow\footnote{For
  an introductory exposition of gradient flows on Riemannian
  manifolds, see, e.g.,~\cite{JMLee02}.} was investigated
\begin{equation}
  R^T \dot{R} = \skew(R^TD)
  \quad\isequivto\quad
  \dot{R} = -{\rm grad}\left(\frac{1}{2}\hsnorm{RD - \id}^2\right)\;.\label{eq:biot_flow}
\end{equation}
The flow \eqref{eq:biot_flow} converges to $R = \id$ for appropriate
initial conditions which is consistent with Grioli's theorem;
cf.~\secref{sec:intro}. Similarly, one can study the gradient
flow for the energy $\frac{1}{2}\hsnorm{\sym(RD - \id)}^2$ given by
\begin{equation}
  \label{eq:sym_flow}
  R^T \dot{R} = -\frac{1}{2} \skew\left((R^TD - \id)^2\right)
  \quad\isequivto\quad
  \dot{R} = -{\rm grad}\left(\frac{1}{2}\hsnorm{\sym(RD - \id)}^2 \right)\;.
\end{equation}
Our present results on critical points of $W(R\,;D)$ determines
the possible asymptotic solutions for the gradient flow
\eqref{eq:sym_flow}. A characterization of \emph{local}
minimizers is currently missing. For example, it is not clear
whether every local minimizer is automatically a global minimizer
which holds in dimension $n = 2$. It seems likely,
  that this holds in $n = 3$ as well. The classification of
local extrema of $W(R\,;D)$ is a completely open question in
$n \geq 4$.

\textbf{Acknowledgments:} Lev Borisov was partially supported by
NSF grant DMS-1201466. Andreas Fischle was supported by German
Research Foundation (DFG) grant SA2130/2-1 and, previously,
partially supported by DFG grant NE902/2-1 (also: SCHR570/6-1).
 \countres

\addcontentsline{toc}{section}{References}
\bibliographystyle{plain}

{\footnotesize
  \setlength{\bibsep}{1pt}
  \bibliography{./bofine-2016-cosseratnd-proof-arxiv-v2}
}

\end{document}